\documentclass[12pt,english,leqno]{amsart}
\usepackage{lmodern}
\usepackage{lmodern}
\usepackage[T1]{fontenc}
\usepackage[utf8]{inputenc}
\usepackage[a4paper]{geometry}
\usepackage{babel}
\usepackage{mathtools}
\usepackage{amstext}
\usepackage{amsthm}
\usepackage{amssymb}
\usepackage{graphicx}
\usepackage[unicode=true,
 bookmarks=false,
 breaklinks=false,pdfborder={0 0 1},backref=section,colorlinks=false]
 {hyperref}

\makeatletter
\numberwithin{equation}{section}
\numberwithin{figure}{section}
\theoremstyle{plain}
\newtheorem{thm}{\protect\theoremname}
  \theoremstyle{plain}
  \newtheorem{cor}[thm]{\protect\corollaryname}
  \theoremstyle{plain}
  \newtheorem{lem}[thm]{\protect\lemmaname}
  \theoremstyle{plain}
  \newtheorem{prop}[thm]{\protect\propositionname}
  \theoremstyle{definition}
  \newtheorem{defn}[thm]{\protect\definitionname}
  \theoremstyle{remark}
  \newtheorem{rem}[thm]{\protect\remarkname}

\makeatother

  \providecommand{\corollaryname}{Corollary}
  \providecommand{\definitionname}{Definition}
  \providecommand{\lemmaname}{Lemma}
  \providecommand{\propositionname}{Proposition}
  \providecommand{\remarkname}{Remark}
\providecommand{\theoremname}{Theorem}

\begin{document}
\global\long\def\C{\mathbb{C}}
\global\long\def\Cd{\C^{\delta}}
\global\long\def\Cprim{\C^{\delta,\circ}}
\global\long\def\Cdual{\C^{\delta,\bullet}}
\global\long\def\Od{\Omega^{\delta}}
\global\long\def\Oprim{\Omega^{\delta,\circ}}
\global\long\def\Odual{\Omega^{\delta,\bullet}}

\global\long\def\P{\mathsf{\mathbb{P}}}
 \global\long\def\E{\mathsf{\mathbb{E}}}
 \global\long\def\sF{\mathcal{F}}
 \global\long\def\ind{\mathbb{I}}

\global\long\def\R{\mathbb{R}}
 \global\long\def\Z{\mathbb{Z}}
 \global\long\def\N{\mathbb{N}}
 \global\long\def\Q{\mathbb{Q}}

\global\long\def\C{\mathbb{C}}
 \global\long\def\Rsphere{\overline{\C}}
 \global\long\def\re{\Re\mathfrak{e}}
 \global\long\def\im{\Im\mathfrak{m}}
 \global\long\def\arg{\mathrm{arg}}
 \global\long\def\i{\mathfrak{i}}
\global\long\def\eps{\varepsilon}
\global\long\def\lamb{\lambda}
\global\long\def\lambb{\bar{\lambda}}

\global\long\def\D{\mathbb{D}}
 \global\long\def\H{\mathbb{H}}
\global\long\def\F{\mathcal{F}}
\global\long\def\outside{\mathcal{\text{out}}}
\global\long\def\winding{\mathrm{wind}}

\global\long\def\dist{\mathrm{dist}}
 \global\long\def\reg{\mathrm{reg}}
 \global\long\def\ii{\mathrm{i}}

\global\long\def\half{\frac{1}{2}}
 \global\long\def\sgn{\mathrm{sgn}}
\global\long\def\Conf{\mathrm{Conf}}
\global\long\def\ZLT{\mathrm{Z_{LT}}}
\global\long\def\Wind{\mathrm{winding}}
\global\long\def\Edges{\mathrm{\mathcal{E}}}
\global\long\def\Dual#1{#1^{\bullet}}

\global\long\def\bdry{\partial}
\global\long\def\pa{\partial}
\global\long\def\en{\epsilon}
\global\long\def\TwiLap#1{\Delta^{#1}}
\global\long\def\oo#1{\mathfrak{o}(#1)}
\global\long\def\EveSub#1{\mathcal{E}(#1)}
 \global\long\def\cl#1{\overline{#1}}
\global\long\def\ZFK{Z_{\text{FK}}}
\global\long\def\Clusters{\mathcal{C}}
\global\long\def\Odp{\hat{\Omega}^{\delta}}
\global\long\def\Edges{\mathcal{E}}
\global\long\def\edges#1{\mathcal{\mathrm{Edges}}(#1)}
\global\long\def\abs#1{\left|#1\right|}
\global\long\def\Vertices{\mathcal{V}}
\global\long\def\bcond{\mathcal{\beta}}
\global\long\def\free{\mathcal{\mathrm{free}}}
\global\long\def\wired{\mathrm{wired}}
\global\long\def\Medial{\mathcal{M}}
\global\long\def\Edgesprim{\hat{\Edges}}
\global\long\def\Conf{\mathrm{Conf}}
\global\long\def\E{\ensuremath{\text{\ensuremath{\mathbb{E}}}}}
\global\long\def\FinLat{\Lambda^{\delta}}
\global\long\def\ModPar{\tau}
\global\long\def\T{\mathbb{T}}
\global\long\def\torus{\mathbb{T}^{\delta}}
\global\long\def\HTECon{\text{const}}
\global\long\def\CriPar{\alpha}
\global\long\def\mesh{\delta}
\global\long\def\wind{\text{wind}}
\global\long\def\KW{\mathcal{KW}}
\global\long\def\OriEdg#1{\vec{\mathcal{E}}(#1)}
\global\long\def\ParFunIsi{Z}
\global\long\def\InvTem{\beta}
\global\long\def\spin#1{\sigma_{#1}}
\global\long\def\vertexes#1{\left|#1\right|}
\global\long\def\ParFunIsiRed{Z^{I}}
\global\long\def\TwiKacWarWeiDetFac#1#2{v(#1,#2)}
\global\long\def\indicator{\mathbb{I}}
\global\long\def\TwiKacWarForShiVec#1{s_{(#1)}}
\global\long\def\unitmat{Id}
\global\long\def\TwiKacWar#1{\mathcal{KW}^{#1}}
\global\long\def\pvec#1{\vec{#1}}
\global\long\def\TorLat{\Lambda}
\global\long\def\DisSpiCon#1{\mathcal{C}\left(#1\right)}
\global\long\def\DisSpiObsVal#1#2{B^{(#1)}(#2)}
\global\long\def\sfrac#1#2{#2/#1}
\global\long\def\TwiParFunWei#1#2{Z^{(#1)}(#2)}
\global\long\def\der#1{\frac{\mathrm{d}}{\mathrm{d}#1}}
\global\long\def\loopfactor#1{(-1)^{q_{#1}(\xi)}}
\global\long\def\HorEdg{e_{H}}
\global\long\def\VerEdg{e_{V}}
\global\long\def\SumEneDen{\E(\en_{H})+\E(\en_{V})}
\global\long\def\Cgr{\mathcal{C}}
\global\long\def\Tdual{\mathbb{T}^{\delta,\star}}
\global\long\def\Tdbl{\hat{\mathbb{T}}}
\global\long\def\Tcvr{\tilde{\T}}
\global\long\def\proj#1#2{\mathrm{Proj}_{#1}\left(#2\right)}
\global\long\def\anw{a_{L}}
\global\long\def\ase{a_{R}}
\global\long\def\DifEneDen{\E(\epsilon_{H})-\E(\epsilon_{V})}
\global\long\def\dzmone{P}
\global\long\def\ConLimSpiObs#1#2{f_{#1}(#2)}
\global\long\def\JacThe#1#2{\theta_{#1}(#2)}
\global\long\def\DisAuxFun#1#2{G_{(#1)}(#2)}
\global\long\def\ConAuxFun#1#2{g_{(#1)}(#2)}
\global\long\def\CorLabPar{\lambda}
\global\long\def\EllMod{k}
\global\long\def\EllMod{k}
\global\long\def\OO#1{\mathcal{O}\left(#1\right)}
\global\long\def\DisSpiObsDivTwiParFun#1#2{\frac{F_{#1}(#2)}{Z_{#1}}}
\global\long\def\hex{\text{hex}}
\global\long\def\tri{\text{tri}}
\global\long\def\En#1#2{\mathcal{E}_{#1}^{(#2)}}
\global\long\def\res{\mathrm{res}}
\global\long\def\fmp#1#2#3{f_{#2}^{(#1)}(#3)}
\global\long\def\Pf{\mathrm{Pf}\,}
\global\long\def\ccor#1{\mathrm{\langle}#1\rangle}
\global\long\def\Op{\mathcal{O}}
\global\long\def\Ceps{C_{\en}}
\global\long\def\PartFun#1{\mathcal{Z}^{(#1)}}
\global\long\def\PFTotal{\mathcal{Z}}
\global\long\def\ds{\mathrm{ds}}
\global\long\def\ns{\mathrm{ns}}
\global\long\def\cs{\mathrm{cs}}
\global\long\def\Conk#1{C_{#1}}
\global\long\def\Ede#1#2{\mathcal{E}_{#1}(#2)}

\title{Energy correlations in the critical Ising model on a torus}

\author{Konstantin Izyurov, Antti Kemppainen and Petri Tuisku}
\begin{abstract}
We compute rigorously the scaling limit of multi-point energy correlations
in the critical Ising model on a torus. For the one-point function,
averaged between horizontal and vertical edges of the square lattice,
this result has been known since the 1969 work of Ferdinand and Fischer.
We propose an alternative proof, in a slightly greater generality,
via a new exact formula in terms of determinants of discrete Laplacians.
We also compute the main term of the asymptotics of the difference
$\E(\en_{V}-\en_{H})$ of the energy density on a vertical and a horizontal
edge, which is of order of $\delta^{2}$, where $\delta$ is the mesh
size. The observable $\en_{V}-\en_{H}$ has been identified by Kadanoff
and Ceva as (a component of) the stress-energy tensor. 

We then apply the discrete complex analysis methods of Smirnov and
Hongler to compute the multi-point correlations. The fermionic observables
are only periodic with doubled periods; by anti-symmetrization, this
leads to contributions from four ``sectors''. The main new challenge
arises in the doubly periodic sector, due to the existence of non-zero
constant (discrete) analytic functions. We show that some additional
input, namely the scaling limit of the one-point function and of relative
contribution of sectors to the partition function, is sufficient to
overcome this difficulty and successfully compute all correlations. 
\end{abstract}

\maketitle

\section{Introduction}

The Ising model is a very famous and influential model in statistical
physics, mathematical physics, discrete mathematics and computer science.
Originally introduced by Wilhelm Lenz in 1920 and named after Ernst
Ising who in 1924 solved the $1$D case of the model \cite{ISING},
the model was initially introduced to study the magnetic phase change
at the so-called Curie temperature. For example, a piece of iron loses
its ferromagnetic properties and becomes a paramagnet at $770$ $^{\circ}$C.

Since its inception, the Ising model has become a very much studied
and archetypal model, a ``test laboratory'', of statistical physics.
This is because the model is rather simple, but still encompasses
a lot of the interesting behavior. Because of the wide applicability
of the model, the Ising model has been studied in many academic fields
ranging from pure mathematics via physics and chemistry to biology
and economics.

The most interesting aspect of Ising model is the fact that it has
an \emph{order-disorder phase transition} in dimension $2$ and higher
dimensions. This fact was established by Peierls in 1936 \cite{PEIRLS}.
The temperature of the phase transition, called \emph{critical temperature},
was predicted by Kramers and Wannier in 1941 \cite{KRAMERSWANNIER}.
After it was understood that the model does have a phase transition
in dimension $2$, finding an exact formula for the partition function
of the model in this case became a central question of much interest
in statistical physics. The feat was achieved by Lars Onsager in a
seminal paper in 1944 \cite{ONSAGER} using transfer matrices. Due
to this exact and rigorous formula being established, it is often
said that the Ising model is \emph{exactly solvable} or \emph{integrable}
in $2$ dimensions. After 1944, the transfer matrix technique was
developed further by Onsager and Kaufman \cite{KAUFMAN,KAUFMANONSAGER}.
The spinor analysis of Kaufman \cite{KAUFMAN} underlies much of our
work.

Besides the algebraic transfer matrix techniques (whose use has continued
and developed since the 1940s), Ising model has also been studied
by so-called combinatorial approach. This approach was advanced by,
among others, Van der Waerden who developed the high temperature expansion
in 1941 \cite{VANDERWAERDEN}, Kac and Ward who introduced the Kac–Ward
determinant in 1952 \cite{KACWARD} (see also \cite{POTTSWARD}),
Sherman, Hurst-Green, Kasteleyn and Fisher who all developed a technique
known as the Pfaffian method in 1960s \cite{SHERMAN} \cite{HURSTGREEN},
\cite{KASTELEYN1,KASTELEYN2,KASTELEYN3}, \cite{FISHER1,FISHER2}
and McCoy and Wu, who wrote the book \cite{MCCOYWU} summarizing the
development of the combinatorial approach in 1970. The results of
Van der Waerden and Kac–Ward especially are of key importance to our
work. The conjecture of Kac and Ward, namely that the Ising model
partition function can be expressed as a determinant of a suitably
chosen matrix, the Kac–Ward matrix, was established mathematically
in 1999 by Dolbilin, Zinov'ev, Mishchenko, Shtan'ko and Shtogrin \cite{DOLBILIN}
and generalized to the form we apply in this paper by David Cimasoni
\cite{CIMASONI1}, \cite{CIMASONI2}.

In 1960s-1970s, another, non-rigorous way to understand Ising model
appeared in the physical literature, the renormalization group. This
approach postulates that the Ising model has a ``continuum limit''
described by a ``quantum field theory (QFT)'' \cite{ALVAREZ-GAUMEMOOREVAFA}.
This theory is thought to be the free fermionic theory. The QFT approach
led to multiple seminal hypothesis being proposed by physicists, among
them the idea that the continuum limit correlations should be related
to the determinants of Laplacians and Dirac operators. We shall show
that for the energy density (correlation of $2$ neighboring spins)
an identity in the discrete setting exists that shows this relation
to determinants of Laplacians. It is noteworthy that this identity
is exact and rigorous and exists even in the discrete setting, not
only in the continuum limit as proposed by QFT approach.

The QFT approach was developed further by Belavin, Polyakov, Zamolodchikov
\cite{BELAVINPOLYAKOVZAMOLODCHIKOV} who in 1984 suggested that the
limiting QFT for the critical Ising model case (when the fermion theory
is ``massless'') has conformal symmetry (CFT)(see also \cite{ALVAREZ-GAUMEMOOREVAFA}).
This lead to the famous prediction that the critical Ising model possesses
conformal covariance in the scaling limit, a conjecture that has inspired
a lot of work in the mathematical community in the last decade. 

In the mathematical community, a lot of work has been done to prove
the conjectures by physicists, and to further the understanding of
the Ising model. Let us here focus only on the study on conformal
covariance of the critical Ising model, that is, the conformal covariance
of its various correlations. The calculation of the full-plane energy
correlations was achieved by Boutillier and de Tilière  \cite{BOUTILLIERDETILIERE1,BOUTILLIERDETILIERE2};
this was later extended to some non-integrable models \cite{giuliani2012scaling,antinucci2020energy}.
The full-plane spin correlations were calculated by Palmer \cite{PALMER}.
Smirnov \cite{SMIRNOV} introduced a powerful tool to analyze the
scaling limit of the Ising model on arbitrary domains, namely, the
discrete holomorphic fermionic observables and associated Riemann
boundary value problem. This was applied to the calculation of the
energy density in simply connected planar domains by Hongler and Smirnov
\cite{HONGLERSMIRNOV}; later extended by Hongler to multi-point energy
correlations in his Ph. D. thesis \cite{hongler_thesis}. The case
of spin correlations in planar domains was solved by Chelkak, Hongler
and Izyurov \cite{CHELKAKHONGLERIZYUROV} and recently these results
were extended to all mixed correlations of primary fields in \cite{CHI_Mixed}. 

The success of the discrete holomorphicity techniques in these papers
naturally leads to the question of whether they can be applied to
the analysis of the model on Riemann surfaces. On the positive side,
the existing universality results \cite{CHELKAKSMIRNOV,ChelkakIzyurovMahfouf,chelkak2020ising}
allow one to extend the techniques to families of graphs that are
flexible enough to approximate any Riemann surface. However, there
are several known difficulties on this path. First, the fermionic
correlations are not well defined on the Riemann surface but have
non-trivial monodromy properties. Thus, one needs to consider $2^{2g}$
observables on a genus $g$ surface where just one was sufficient
in the planar case. This is similar to the $2^{2g}$ Kac–Ward determinants
introduced by Cimasoni \cite{CIMASONI1,CIMASONI2}, and in fact this
is no coincidence since the observables are related to the inverses
of the Kac–Ward matrices, see \cite{lis2014fermionic,CHELKAKCIMASONIKASSEL}.
What's more, the observables do not compute the energy or spin correlations
per se, but rather their correlations with one of $2^{2g}$ ``topological''
observables, normalized by the expectation of that observable, given
by the corresponding Kac–Ward determinant. Thus, to recover the correlations
of interest, one would need to compute scaling limit of the ratio
of the square root of each of the $2^{2g}$ Kac-Ward determinants
to their sum. Finally, a degeneracy occurs in that one of the observables
happens to be a non-zero constant, whose value is not immediately
recoverable by the discrete complex analysis methods.

The present paper provides a first step in this program, treating
the case of energy correlations on a torus. The advantage of working
with the flat torus is that the above-mentioned difficulties can be
treated by other methods. Thus, the Kac-Ward determinants can be calculated
explicitly, and their limiting asymptotics analyzed. Moreover, as
the partition function can be computed for any temperature, one can
differentiate it with respect to the temperature to obtain the average
energy density at criticality. This has been done by Ferdinand and
Fischer \cite{FERDINANDFISHER} and later refined in \cite{SALAS1,SALAS2,izmailian2002exact}
by using the Onsager–Kaufman expression for the partition function.
It turns out that a slight refinement of this result (separating the
vertical and the horizontal edges) is sufficient to remedy the above-mentioned
degeneracy problem. We propose an alternative computation of the average
energy density, based instead on the Kac–Ward solution. The advantage
of our approach is that we relate, at the discrete level, the average
energy density to a ratio of determinants of the discrete Laplacians.
The asymptotics of the determinants of Laplacians has been recently
analyzed in great generality \cite{finski2020finite,finski2020spanning,izyurov2020asymptotics}.
Thus, if a similar ``bosonization'' relation is found on Riemann
surfaces, the asymptotics of the one-points function in the scaling
limit would be readily available. 

It turns out that once the above-mentioned difficulties are treated,
or the missing pieces supplied as an input, the rest of the computation
of the arbitrary many-point energy correlations can be done by the
discrete complex analysis methods, recovering physicists' predictions
in this regard, see Theorem \ref{thm: multipoint} below. While for
even number of marked points, the argument in \cite{hongler_thesis,CHI_Mixed}
can be extended almost verbatim, for odd number points we had to modify
it, as we have found no combinatorial counterpart for the ``propagator''
$\zeta(e_{n}-e_{m})$ in (\ref{eq: corr_en_odd}) below. Therefore,
rather than reducing everything to the two-point fermionic observables
as in \cite{hongler_thesis,CHI_Mixed}, we had to work out the convergence
result directly for a multi-point observable, and only obtain the
Pfaffian formula (\ref{eq: corr_en_odd}) by analysing the resulting
scaling limit in the continuum.

\subsubsection*{Acknowledgements. }

Work supported by the Academy of Finland via Centre of Excellence
in Analysis and Dynamics research and the academy project ``Critical
phenomena in dimension two: analytic and probabilistic methods''.
We are grateful to Antti Kupiainen and Dmitry Chelkak for useful discussions, and to the anonymous referee for careful reading of the manuscript and many useful suggestions.
We thank David Loeffler for pointing out a quick way to derive the
Kronecker limit formula with anti-periodic boundary conditions, used
in the proof of Corollary \ref{THEOREM1}.

\section{Main results}

\subsection{Setup and notation}

We will study the Ising model on tori $\T^{\delta}:=\mesh\Z^{2}/\FinLat$,
where $\delta>0$ is the \emph{mesh size}, $\Lambda^{\delta}$ is the lattice
\[
\Lambda^{\delta}=\{n\omega_{1}^{\mesh}+m\omega_{2}^{\mesh}:n,m\in\Z\}
\]
and $\omega_{1}^{\mesh},\omega_2^{\mesh}\in\mesh\Z^{2}$ are non-collinear. We will
be interested in the \emph{scaling limit} of the model where $\mesh\to0$,
$\omega_{i}^{\mesh}\to\omega_{i}\in\C\setminus\{0\}$ for $i=1,2$. We will
assume that the \emph{modular parameter} $\ModPar:=\omega_{2}/\omega_{1}$
satisfies $\im\ModPar>0$, and denote $\T=\C/\Lambda$, where $\Lambda=\{n\omega_{1}+m\omega_{2}:n,m\in\Z\}$.

The (zero magnetic field) Ising model on $\torus$ is the probability
measure on spin configurations $\sigma:\torus\to\{\pm1\}$ given by
\begin{align*}
\P\left[\sigma\right]=\frac{1}{\ParFunIsi}\exp\left(\InvTem\sum_{x\sim y}\spin x\spin y\right)
\end{align*}
where the sum is over all pairs of nearest-neighboring vertices of
$\T^{\delta}$, $\InvTem>0$ is a parameter called the \emph{inverse
temperature}, and 
\begin{align*}
\ParFunIsi=\sum_{\sigma:\torus\to\{\pm1\}}\exp\left(\InvTem\sum_{x\sim y}\spin x\spin y\right)
\end{align*}
is the \emph{partition function} of the model, which ensures that
the configuration probabilities sum up to $1$. Our main results will
concern the \emph{critical temperature} $\InvTem=\InvTem_{c}=\frac{1}{2}\log(\sqrt{2}+1)$.
We will denote by $\P$ and $\E$, respectively, the probability and
the expectation with respect to the above measure.

The main object of our interest is the \emph{energy
observable}. If $(xy)$ is an edge of $\T^{\delta}$, we denote 
\[
\en_{(xy)}=\sigma_{x}\sigma_{y}-\frac{1}{\sqrt{2}}.
\]
The constant $\frac{1}{\sqrt{2}}$ is the expectation of $\sigma_{x}\sigma_{y}$
in the \text{full-plane}, i.e., the thermodynamic limit. Thus,
our results will measure how the toric boundary conditions affect
the expectation and the correlations of $\en_{(xy)}$. Clearly, since
$\torus$ carries an action of $\delta\Z^{2}$ by translations, we
have $\E\en_{e}=\E\en_{\hat{e}}$ if the edges $e,\hat{e}$ are either
both horizontal, or both vertical. We thus denote by $\E\en_{H}$
and $\E\en_{V}$ the expectation of $\en$ on any horizontal and vertical
edge, respectively.

\subsection{Main results}

\label{subsec: main_results}Our first result concerns the sum of
vertical and horizontal energy densities. We introduce some notation.
Given $\omega_{1,2}^{\mesh}$ as above and $i,j\in\Z_{2}$, denote
\[
V_{\delta}^{ij}:=\{f:\delta\Z^{2}\to\R:f(v+\omega_{1}^{\delta})\equiv(-1)^{i}f(v),\;f(v+\omega_{2}^{\delta})\equiv(-1)^{j}f(v)\}.
\]
Each of $V_{\delta}^{ij}$ is a linear space of dimension $|\torus|$;
$V_{\delta}^{00}$ can be viewed just as the set of functions on $\torus$.
Note that the lattice Laplacian $\Delta_{\delta}f(x):=\sum_{y\sim x}(f(y)-f(x))$
preserves each of these spaces. We denote by $\Delta_{\delta}^{ij}$
the restriction of $\Delta$ to $V_{\delta}^{ij}$.
\begin{thm}
\label{thm: sum} For the critical Ising model on $\torus$,
we have, for a horizontal edge $H$ and a vertical edge $V$: 
\begin{equation}
\E\en_{V}+\E\en_{H}=4\cdot\frac{\sqrt{\det^{\star}\TwiLap{00}_{\delta}}}{\sqrt{\det\TwiLap{10}_{\delta}}+\sqrt{\det\TwiLap{01}_{\delta}}+\sqrt{\det\TwiLap{11}_{\delta}}}\cdot\frac{1}{|\torus|},\label{eq: main_combinatirual}
\end{equation}
where $\det^{\star}$ denotes the product of all non-zero eigenvalues, and $|\torus|$ denotes the number of vertices of $\torus$. 
\end{thm}

The continuous counterpart of this identity has appeared in the CFT
literature, see \cite{diFrancesco1987critical}, \cite{FELDER}, \cite{DIFRANCESCOMATHIEUSENECHAL},
however, the discrete version is, to the best of our knowledge, new.
Cimasoni \cite{CIMASONI2} has related the determinants of the discrete
Laplacians to the critical Ising partition functions on arbitrary
isoradial graphs embedded on a torus. However, we were unable to adapt
his methods to the computation of energy densities. A more general
approach to bosonization of the Ising model was developed by Dubédat
in \cite{dubedat2011exact}, however, it involves more complicated
modifications of the original graph and does not seem to lead to (\ref{eq: main_combinatirual})
either. In Section \ref{sec: triangular}, we provide an analog of
(\ref{eq: main_combinatirual}) for triangular lattice; we do not
know whether such analogs hold true for other lattices.

By combining the above formula with known results on the asymptotics
of determinants of discrete Laplacians, we recover the asymptotics
result of Ferdinand and Fischer \cite{FERDINANDFISHER}, in a slightly
greater generality of arbitrary torus as compared to diagonal one:
\begin{cor}
\label{THEOREM1} In the limit $\delta\to0$, $\omega_{1,2}^{\delta}\to\omega_{1,2}$
with $\omega_{2}/\omega_{1}=\tau$, we have 
\begin{align*}
\E\en_{V}+\E\en_{H}=\frac{2(\im\tau)^{\frac{1}{2}}|\theta_{2}\theta_{3}\theta_{4}|}{|\theta_{2}|+|\theta_{3}|+|\theta_{4}|}\cdot\frac{1}{|\T|^{\frac{1}{2}}}\cdot\delta+\oo{\delta}.
\end{align*}
Hereinafter, our notation for the theta constants is $\theta_{i}:=\theta_{i}(\tau):=\theta_{i}(0,q)$,
where $q=e^{\i\pi\tau}$ and $\theta_{i}(0,q)$ is as in \cite[Chapter 20]{NIST:DLMF}, and $|\T|$ denotes the area of $\T$.
\end{cor}

Our second result concerns the asymptotics of the \emph{difference}
between energy density on vertical and horizontal edges. This difference
happens to be of order $\delta^{2};$ such an observable was identified
\cite{KadanoffCeva} as a component of \emph{stress-energy tensor}
in the model, see also \cite{ChelkakGlazmanSmirnov}. 
\begin{thm}
\label{thm: difference} In the limit $\delta\to0$, $\omega_{1,2}^{\delta}\to\omega_{1,2}$
with $\omega_{2}/\omega_{1}=\tau$, we have 
\begin{equation}
\DifEneDen=\frac{\sqrt{2}\pi}{24}\cdot H(\omega_{1},\omega_{2})\cdot\mesh^{2}+\oo{\mesh^{2}},\label{eq: diff}
\end{equation}
where $H(\omega_{1},\omega_{2})$ is given by 
\[
\frac{\PartFun{01}}{\PFTotal}\re\left[\omega_{1}^{-2}(\theta_{2}^{4}-2\theta_{3}^{4})\right]+\frac{\PartFun{10}}{\PFTotal}\re\left[\omega_{1}^{-2}(\theta_{2}^{4}+\theta_{3}^{4})\right]+\frac{\PartFun{11}}{\PFTotal}\re\left[\omega_{1}^{-2}(\theta_{3}^{4}-2\theta_{2}^{4})\right]
\]
where 
\begin{equation}
\PartFun{01}=\left|\theta_{2}\right|,\quad\PartFun{10}=\left|\theta_{4}\right|\quad\PartFun{11}=\left|\theta_{3}\right|,\label{eq: Z}
\end{equation}
and $\PFTotal=$ $\PartFun{01}+\PartFun{10}+\PartFun{11}.$ 
\end{thm}

We believe that his result could have been obtained by the methods
of \cite{FERDINANDFISHER,SALAS1,SALAS2,izmailian2002exact}, by writing
down the partition function of the anisotropic Ising model and then
differentiating with respect to the coupling constant separately on
vertical and horizontal edges. However, we obtain it as a very simple
by-product of our analysis of discrete holomorphic fermionic observables. 

We record a corollary that will be useful in the study of multi-point
energy correlations:
\begin{cor}
\label{Cor: hor_only} In the limit $\delta\to0$, $\omega_{1,2}^{\delta}\to\omega_{1,2}$
with $\omega_{2}/\omega_{1}=\tau$, we have 
\begin{align*}
\E\en_{H}=\frac{(\im\tau)^{\frac{1}{2}}|\theta_{2}\theta_{3}\theta_{4}|}{|\theta_{2}|+|\theta_{3}|+|\theta_{4}|}\cdot\frac{1}{|\T|^{\frac{1}{2}}}\cdot\delta+\oo{\delta}.
\end{align*}
and similarly for $\E\en_{V}$. 
\end{cor}

In fact, higher-order terms (up to $\delta^{3}$) of the expansion
of $\E\en_{H}+\E\en_{V}$ were computed by Salas and Izmailyan–Hu
in \cite{SALAS1,izmailian2002exact}. They showed that the $\delta^{2}$
term is absent from the expansion. Therefore, our results in fact
give the expansion of $\E\en_{H}$ up to order $\oo{\delta^{2}}$.
Chinta, Jorgenson and Karlsson \cite{CHINTAJORGENSONKARLSSON1} indicate
a way to compute the asymptotic expansion of $\det^{\star}\Delta^{00}$
up to arbitrary order in $\delta$. Combined with our Theorem \ref{thm: difference}
, this could in principle used to the compute $\E\en_{H}+\E\en_{V}$
up to arbitrary order.

We also compute the scaling limit of multi-point correlation functions. 
\begin{thm}
\label{thm: multipoint}In the scaling limit $\delta\to0$, $\T^{\delta}\to\T$,
as $e_{1},\dots,e_{k}$ approach distinct points of $\T$, we have,
for even $k,$ 
\begin{multline}
\pi^{k}\delta{}^{-k}\E\left[\en_{e_{1}}\dots\en_{e_{k}}\right]\longrightarrow\frac{\PartFun{01}}{\PFTotal}\left|\Pf\left[\cs_{\omega_{1},\omega_{2}}(e_{n}-e_{m})\right]\right|^{2}\\
+\frac{\PartFun{10}}{\PFTotal}\left|\Pf\left[\ns_{\omega_{1},\omega_{2}}(e_{n}-e_{m})\right]\right|^{2}+\frac{\PartFun{11}}{\PFTotal}\left|\Pf\left[\ds_{\omega_{1},\omega_{2}}(e_{n}-e_{m})\right]\right|^{2},\label{eq: corr_en_even}
\end{multline}
where \textup{$\cs_{\omega_{1},\omega_{2}},$} $\ns_{\omega_{1},\omega_{2}},$
$\ds_{\omega_{1},\omega_{2}}$ are Jacobian elliptic functions, see
Section \ref{sec: Scaling-limits }.

For odd $k$, we have 
\begin{equation}
\pi^{k}\delta{}^{-k}\E\left[\en_{e_{1}}\dots\en_{e_{k}}\right]\to\i^{k}\cdot\Pf M,\label{eq: corr_en_odd}
\end{equation}
 where $M$ is the $2k\times2k$ anti-symmetric matrix given by $M_{2n-1,2m-1}=\zeta_{\omega_{1},\omega_{2}}(e_{n}-e_{m}),$
$M_{2n,2m}=\overline{\zeta_{\omega_{1},\omega_{2}}(e_{n}-e_{m})},$
and 
\[
M_{2n-1,2m}\equiv-\pi\i\cdot\frac{(\im\tau)^{\frac{1}{2}}|\theta_{2}\theta_{3}\theta_{4}|}{|\theta_{2}|+|\theta_{3}|+|\theta_{4}|}\cdot\frac{1}{|\T|^{\frac{1}{2}}}.
\]
Here $\zeta_{\omega_{1},\omega_{2}}$ is the Weierstrass $\zeta$–function,
see Section \ref{sec: Scaling-limits }.
\end{thm}

The formula (\ref{eq: corr_en_even}) has been predicted in the physics
literature by di Francesco, Saleur and Zuber \cite{diFrancesco1987critical},
see also \cite[Section 12]{DIFRANCESCOMATHIEUSENECHAL}. The formula
(\ref{eq: corr_en_odd}), on the other hand, appears to be new. We
expect it to be related to the prediction of \cite{diFrancesco1987critical}
by an appropriate version of Fay's formula.

\section{Partition functions and Kac–Ward determinants}

\label{sec: Kac-Ward}

In this Section, we record the necessary results involving Kac–Ward
solution to the critical Ising model. This approach was originated
in \cite{KACWARD}; the first complete proof was given in \cite{DOLBILIN},
and in the case of surfaces in \cite{CIMASONI1}, with simplified
proof in \cite{CHELKAKCIMASONIKASSEL}. We now recall the required
material in detail in the case of a torus, following \cite[Section 4]{CHELKAKCIMASONIKASSEL}.

Let us denote $\EveSub{\torus}$ the set of even subgraphs of $\torus$
(understood as subsets of edges of $\torus$), that is, subgraphs
of $\torus$ such that each vertex is adjacent to an even number of
edges. It is well known (see for example \cite[Subsection 2.1]{CIMASONI2}
for a short exposition) that the partition function $\ParFunIsi$
can be expressed as (the \emph{high temperature expansion}) 
\begin{align*}
\ParFunIsi=\cosh\left(\InvTem\right)^{\abs{\edges{\torus}}}2^{\vertexes{\torus}}\sum_{\xi\in\EveSub{\torus}}\tanh\left(\InvTem\right)^{\abs{\xi}}.
\end{align*}
To avoid the constant factor appearing in every formula, let us
define the partition function of the Ising model without this factor:
\begin{align}
\ParFunIsiRed:=\sum_{\xi\in\EveSub{\torus}}\CriPar^{\abs{\xi}},\quad \alpha=\tanh(\beta)\label{HIGHTEMPERATUREEXPANSIONREDUCED}
\end{align}
We shall refer to this sum also as ``partition function''.

Pick a point $z_{0}\in\C$ such that the two lines $\gamma_{1,2}:=\{z_{0}+t\omega_{1,2}^{\delta}:t\in\R\}$
do not intersect $\delta\Z^{2}$. We identify $\gamma_{1,2}$ with
their projection onto $\torus$. Given an edge $e\in\torus$ and $i,j\in\Z_{2}$,
put 
\begin{equation}
\varphi_{ij}(e)=i\ind_{e\cap\gamma_{1}\neq0}+j\ind_{e\cap\gamma_{2}\neq0}\,\mod2=i\varphi_{10}(e)+j\varphi_{01}(e)\,\mod2\label{eq:def_phi}
\end{equation}
These are four $\Z_{2}$-valued \emph{flat connections} on $\torus$, i.e., the function $\gamma\mapsto \varphi_{ij}(\gamma):=\sum_{e\in\gamma}\varphi_{ij}(e)$ vanishes identically on contractible paths.
In fact, these are the only flat connections
up to gauge equivalence. As explained in \cite[Section 4]{CHELKAKCIMASONIKASSEL},
this allows one to construct four \emph{spin structures} $\lambda_{ij}$,
which are, roughly speaking, ways to assign a winding number modulo
$4\pi$ to a closed lattice path. Namely, if the lattice path $\gamma$
consist of the edges $e_{1},e_{2},\dots,e_{k}$, then $\wind_{\lambda_{ij}}(\gamma)=\wind(\gamma)+2\pi\varphi_{ij}(\gamma)$,
where $\wind$ is the winding of the lift of $\gamma$ to the plane.
This, in its turn, allows one to define four \emph{quadratic forms}
$q_{ij}$ on $\EveSub{\torus}$: given $\xi\in\EveSub{\torus}$, decompose
it into a collection of loops $\xi_{1},\dots\xi_{N}$ that do not
intersect themselves or each other transversally (to this end, for
each vertex of degree 4 in $\xi$, pick any two incident edges forming
a right angle, and declare them belong to the same loop, and also
other two to belong to the same loop). Then, put $(-1)^{q_{ij}(\xi)}:=\prod_{k=1}^{N}(-\exp(\frac{\i}{2}\wind_{\lambda_{ij}}(\xi_{k})))$.

Let us calculate $q_{ij}(\xi)$ concretely. If a loop $\xi_{k}$ lifts
to a closed loop on $\delta\Z^{2}$, then it crosses $\gamma_{1}$
and $\gamma_{2}$ an even number of times, and thus $\wind_{\lambda_{ij}}(\xi_{k})=\wind(\xi_{k})=2\pi$.
Otherwise, it lifts to a path connecting two distinct points in the
plane, and we have $\wind(\xi_{k})=0.$ Thus, we have 
\[
(-1)^{q_{ij}(\xi)}=(-1)^{N(\xi)+\varphi_{ij}(\xi)},
\]
where $N(\xi)$ is a number of non-contractible loops in (the decomposition
of) $\xi$. The lift of a non-contractible loop connects $z$ and
$z+m_{1}\omega_{1}^{\delta}+m_{2}\omega_{2}^{\delta}$; since the
loops are simple and non-intersecting, $m_{1}$ and $m_{2}$ are relatively
prime and the same for all non-contractible loops in $\xi$. This means
that $N(\xi)=\varphi_{01}(\xi)+\varphi_{01}(\xi)-\varphi_{01}(\xi)\varphi_{10}(\xi)$,
by exclusion-inclusion: configurations with an odd number of non-contractible
loops contributes to $\varphi_{01}(\xi)$ iff $m_{1}$ is odd (respectively,
to $\varphi_{10}(\xi)$ iff $m_{2}$ is odd). Therefore,
\begin{equation}
(-1)^{q_{ij}(\xi)}=(-1)^{(1-i)\varphi_{10}(\xi)+(1-j)\varphi_{01}(\xi)+\varphi_{10}(\xi)\cdot\varphi_{01}(\xi)}.\label{eq: qf}
\end{equation}
Of course, this is a manifestation of the general fact that $q_{ij}$
is a quadratic form on the homology space $H_{1}(\torus,\Z_{2})$. We summarize the values of $q_{ij}(\chi)$ in Table 1:
\begin{table}
\caption{The values of $q_{ij}$ on the four homology classes in $H_{1}(\torus,\Z_{2})$. Here $\hat{\gamma}_1$ and $\hat{\gamma}_2$ are simple loops in $\torus$ that lift to paths on $\delta\Z^2$ connecting the origin to $\omega_{2}^{\delta}$ and $\omega_{1}^{\delta}$ respectively.}
\begin{center}
\begin{tabular}{|c|c|c|c|c|}
\hline
$\xi$& $q_{00}(\xi)$ & $q_{01}(\xi)$ & $q_{10}(\xi)$ & $q_{11}(\xi)$\\ \hline
$[\emptyset]$& $0$& $0$& $0$& $0$\\ \hline
$[\hat{\gamma}_1]$& $1$& $1$& $0$& $0$\\ \hline
$[\hat{\gamma}_2]$& $1$& $0$& $1$& $0$\\ \hline
$[\hat{\gamma}_1+\hat{\gamma}_2]$& $1$& $0$& $0$& $1$\\ \hline
\end{tabular}
\end{center}
\end{table}

To each of the four spin structures, one associates a twisted Kac–Ward
matrix $\KW^{ij}$. That is is a matrix indexed by the set $\OriEdg{\torus}$
of oriented edges of $\torus$ , by putting $\KW^{ij}:=\unitmat-T^{ij}$,
where 
\begin{align*}
T^{ij}_{\pvec e,\pvec e'}:=\begin{cases}
(-1)^{\varphi_{ij}\left(\pvec e\right)}\exp\left(\frac{\i}{2}\wind\left(\pvec e,\pvec e'\right)\right)\alpha,\text{if }t\left(\pvec e\right)=o\left(\pvec e'\right)\text{ but }\pvec e'\neq-\pvec e\\
0,\text{ otherwise}
\end{cases}
\end{align*}
Here $t(\pvec e)$ and $o(\pvec e)$ denote the end and the beginning
of $e$, respectively.

The following theorem due to Cimasoni relates the Kac–Ward matrices
with Ising partition functions: 
\begin{thm}
\label{LEMMACIMASONILEMMA} We have, for $i,j\in\Z_{2}$, 
\begin{align}
\sqrt{\det\KW^{ij}}=\sum_{\xi\in\EveSub{\torus}}(-1)^{q_{ij}(\xi)}\alpha^{\abs{\xi}}=:Z^{(ij)}.\label{eq: cimasoni}
\end{align}
\end{thm}

The following Lemma is a particular case of the equation (4.6) in
\cite{CHELKAKCIMASONIKASSEL}. It allows, in particular, to express
the partition function $\ParFunIsiRed$ in terms of determinants of
Kac–Ward matrices. 
\begin{lem}
\label{Lemma: sum_qf} We have, for any $\xi\in\EveSub{\torus}$,
\begin{align}
(-1)^{q_{10}(\xi)}+(-1)^{q_{01}(\xi)}+(-1)^{q_{11}(\xi)}-(-1)^{q_{00}(\xi)}=2\label{eq: lem_qf}\\
2\ParFunIsiRed=-\sqrt{\det\KW^{00}}+\sqrt{\det\KW^{01}}+\sqrt{\det\KW^{10}}+\sqrt{\det\KW^{11}}
\end{align}
\end{lem}

\begin{proof}
The first identity easily follows from (\ref{eq: qf}): if $\varphi_{01}(\xi)=\varphi_{10}(\xi)=0$,
then all the terms in the left-hand side are equal to $1$, and altering
$\varphi_{01}(\xi)$ or $\varphi_{10}(\xi)$ always changes the sign
of exactly two terms. The second identity is obtained by multiplying
the first one by $\alpha^{|\xi|}$ and summing over $\xi$. 
\end{proof}
In the case of a torus, the determinant of Kac–Ward matrix can be
calculated explicitly. Let us first denote, for $q\in\C$, $z\left(q\right)=\exp\left(2\pi\i\re q\right)$
and $w\left(q\right)=\exp\left(2\pi\i\im q\right)$, and 
\begin{align}
\TwiKacWarWeiDetFac{\alpha}q:=\left(1+\alpha^{2}\right)^{2}+\alpha\left(\alpha^{2}-1\right)\left(z\left(q\right)+z\left(q\right)^{-1}+w\left(q\right)+w\left(q\right)^{-1}\right);\label{eq: def_v_alpha_q}
\end{align}
note that this is real non-negative for all $\alpha\in (0,1)$ and any $q$, and vanishes only at $q\in\mathbb{Z}$ and $\alpha=\sqrt{2}-1$. For the following (purely combinatorial) discussion, we will assume
$\delta=1$. The dual lattice $\TorLat^{*}$ of a lattice $\TorLat\subset\Z^{2}$
is defined as 
\begin{align*}
\TorLat^{*}:=\{q\in\C\,|\re{z}\re{q}+\im{z}\im{q}\in\Z\textnormal{ for all }z\in\TorLat\},
\end{align*}
and we define shift vectors (complex numbers) $\TwiKacWarForShiVec{ij}$,
$i,j\in\{0,1\}$ as follows: 
\begin{align*}
\re{\TwiKacWarForShiVec{ij}}\re{\omega_{1}^{\delta}}+\im{\TwiKacWarForShiVec{ij}}\im{\omega_{1}^{\delta}} & =\frac{i}{2}\\
\re{\TwiKacWarForShiVec{ij}}\re{\omega_{2}^{\delta}}+\im{\TwiKacWarForShiVec{ij}}\im{\omega_{2}^{\delta}} & =\frac{j}{2}.
\end{align*}
Note that such $\TwiKacWarForShiVec{ij}$ with these properties exist
and are unique as $\omega_{1,2}^{\delta}$ form a basis of the plane.
\begin{thm}
\label{thm: det_KW} One has 
\begin{align}
\det\KW^{ij}=\prod_{q\in\sfrac{\Z^{2}}{\TorLat^{*}}+\TwiKacWarForShiVec{ij}}\TwiKacWarWeiDetFac{\alpha}q.\label{EQUATIONPRODUCTFORMULAFORTHETWISTEDKACWARDDETERMINANTi,j}
\end{align}
 
\end{thm}

\begin{proof}
The theorem is as in \cite[Lemma 4.1]{CIMASONIDUMINIL-COPIN}. For
completeness, let us provide a short argument. This argument
originates from Kac and Ward \cite{KACWARD}. 

Observe that we can think of $\KW^{ij}$ as a non-twisted Kac–Ward
matrix $\KW:=\KW^{00}$ acting on the corresponding space
of (anti-)periodic function of oriented edges, similar to $V^{ij}.$
The untwisted Kac–Ward operator $\KW$ commutes with shifts by $\Z^{2},$
hence, it is natural to look for eigenvectors of $\KW^{ij}$ that
are also eigenvectors of these shifts. The latter in general have
the form
\[
U_{q}(\pvec e)=\hat{U}([\pvec e])e^{2\pi \i\re qx(\pvec e)}e^{2\pi \i\im qy(\pvec e)}=\hat{U}([\pvec e])z(q)^{x(\pvec e)}w(q)^{y(\pvec e)},
\]
where $\hat{U}([\pvec e])$ is some function of the equivalence class
of $\pvec e$ under shifts, $x(\pvec e),y(\pvec e)\in \mathbb{Z}^2\subset \mathbb{R}^2$ are ``coordinates''
of $\pvec e$ (say, of its beginning), and the condition $q\in\sfrac{\Z^{2}}{\TorLat^{*}}+\TwiKacWarForShiVec{ij}$
stems from the (anti-)periodicity requirement. The
action of $\KW^{ij}$ on $U_{q}$ is then straightforward to compute,
\[
\KW(U_{q})(\pvec e)=\left(\KW^{(q)}(\hat{U})\right)([\pvec e])z(q)^{x(\pvec e)}w(q)^{y(\pvec e)},
\]
 where $\KW^{(q)}$ is the twisted Kac–Ward operator of a
one-vertex torus corresponding to the connection $\varphi^{q}$ with
$\varphi^{q}(\pvec e)=z(q)$ for $\pvec e$ pointing to the right,
and $\varphi^{q}(\pvec e)=w(q)$ for $\pvec e$ pointing upwards.
Therefore, we conclude that 
\[
\det\KW^{ij}=\prod_{q\in\sfrac{\Z^{2}}{\TorLat^{*}}+\TwiKacWarForShiVec{ij}}\det\KW^{(q)}.
\]
 Finally, $\KW^{(q)}$ is an explicit $4\times4$ matrix, and a straightforward
computation yields $\det\KW^{(q)}=v(\alpha,q).$
\end{proof}
We observe that the \emph{same} expressions give the determinants
of the discrete Laplacians defined in Section \ref{subsec: main_results}.
\begin{prop}
\label{prop: KW_to_Laplacian} The discrete Laplacian $\Delta^{ij}$
has eigenvalues 
\begin{align}
\frac{1}{2\alpha_{c}^{2}}\TwiKacWarWeiDetFac{\alpha_{c}}q,\quad q\in\sfrac{\Z^{2}}{\TorLat^{*}}+\TwiKacWarForShiVec{ij},\label{eq: detlap_viaKacward}
\end{align}
where $\alpha_{c}=\sqrt{2}-1$. 
\end{prop}

\begin{proof}
Observe that if $q\in\sfrac{\Z^{2}}{\TorLat^{*}}+\TwiKacWarForShiVec{ij}$, then the exponential function $f_{q}(z)=\exp(2\pi\i(\re{q}\re{z}+\im{q}\im{z}))$, where $z\in\mathbb{Z}^2$, belong to $V^{ij}$,  and $\Delta f_{q}=\frac{1}{2\alpha_{c}^{2}}\TwiKacWarWeiDetFac{\alpha_{c}}qf_{q}$.
Thus we find $|\Lambda^{\star}/\Z^{2}|=|\torus|$ distinct eigenvalues;
as $V_{\delta}^{ij}$ has dimension $|\torus|$, these are all the eigenvalues.\footnote{The reader may observe that this is essentially the same proof as that of Theorem 8; the latter is a bit more involved since the fundamental domain has only one vertex and four oriented edges.}
\end{proof}

\section{The sum of the horizontal and vertical edge energy densities}

\label{SECTIONTHESUMOFTHEHORIZONTALANDVERTICALEDGEENERGYDENSITIES}

We continue with a standard Lemma representing energy density in high-temperature
expansion, see e.g. \cite[Subsection 2.1]{CIMASONI2} or \cite{KadanoffCeva} for background.
\begin{lem}
\label{LEMMAAUXILIARYFORMOFENERGYDENSITY} Let $e=(xy)\in\edges{\torus}$
be any edge. For the critical temperature $\alpha=\sqrt{2}-1$, one
has 
\begin{align}
\E(\en_{e})=\frac{1}{\sqrt{2}}\frac{1}{\ParFunIsiRed}\sum_{\xi\subset\EveSub{\torus}}\left(\CriPar^{-1}\indicator_{e\in\xi}-\CriPar\indicator_{e\notin\xi}\right)\CriPar^{\abs{\xi}}.\label{eq: energy_HT}
\end{align}
\end{lem}

\begin{proof}
Denote by $\Ede e{\torus}$ the set of configurations $\omega\subseteq\edges{\torus}$
such that every vertex except the end-vertices of $e$ has even degree
and the end-vertices of $e$ have odd degree in $\omega$. If $\xi\in\EveSub{\torus}$,
define $\omega_{e}(\xi)\in\Ede e{\torus}$ by $\omega_{e}(\xi)=\xi\cup e$
if $e\notin\xi$ and $\omega_{e}(\xi)=\xi\setminus e$ if $e\in\xi$.
Then, $\omega_{e}$ is is a bijection between $\EveSub{\torus}$ and
$\Ede e{\torus}$. By high-temperature expansion, we have 
\begin{align*}
\E(\sigma_{x}\sigma_{y})=\frac{1}{\ParFunIsiRed}\sum_{\omega\in\Ede e{\torus}}\CriPar^{\abs{\omega}}=\frac{1}{\ParFunIsiRed}\sum_{\xi\in\EveSub{\torus}}\CriPar^{\abs{\omega_{e}(\xi)}}=\frac{1}{\ParFunIsiRed}\sum_{\xi\in\EveSub{\torus}}\left(\ind_{e\in\xi}\CriPar^{\abs{\xi}-1}+\ind_{e\notin\xi}\CriPar^{\abs{\xi}+1}\right)\\
=\frac{1}{\sqrt{2}}+\frac{1}{\ParFunIsiRed}\sum_{\xi\in\EveSub{\torus}}\left(\left(\alpha^{-1}-\frac{1}{\sqrt{2}}\right)\ind_{e\in\xi}\CriPar^{\abs{\xi}}+\left(\alpha-\frac{1}{\sqrt{2}}\right)\ind_{e\notin\xi}\CriPar^{\abs{\xi}}\right),
\end{align*}
and we finish by noting that $\alpha^{-1}-\frac{1}{\sqrt{2}}=\frac{\alpha^{-1}}{\sqrt{2}}$
and $\alpha-\frac{1}{\sqrt{2}}=-\frac{\alpha}{\sqrt{2}}$. 
\end{proof}
For $i,j\in\Z_{2}$ and $e$ and edge of $\torus$, define 
\begin{align*}
\DisSpiObsVal{ij}e:=\sum_{\xi\in\EveSub{\torus}}\left(\CriPar^{-1}\indicator_{e\in\xi}-\CriPar\indicator_{e\notin\xi}\right)(-1)^{q_{ij}(\xi)}\CriPar^{\abs{\xi}}.
\end{align*}
In view of Lemma \ref{LEMMAAUXILIARYFORMOFENERGYDENSITY} and (\ref{eq: lem_qf}),
we have 
\begin{align}
\E(\en_{e})=%
\frac{1}{2\sqrt{2}}\frac{1}{\ParFunIsiRed}\left(\DisSpiObsVal{01}e+\DisSpiObsVal{10}e+\DisSpiObsVal{11}e-\DisSpiObsVal{00}e\right).\label{eq: energy_B}
\end{align}

It turns out that at the critical temperature, it is possible to express
the sum of $\DisSpiObsVal{ij}e$ for horizontal and vertical edge
in terms of a determinant of the Laplacian:
\begin{lem}
\label{lem: B_ij}One has, for $\alpha=\alpha_{c}=\sqrt{2}-1,$ 
\begin{align*}
\DisSpiObsVal{ij}{e_{H}}+\DisSpiObsVal{ij}{e_{V}}=\begin{cases}
-2^{\frac{|\torus|+5}{2}}\CriPar^{|\torus|}\sqrt{\mathrm{det}^{\star}\Delta^{00}}\frac{1}{|\torus|}, & i=j=0;\\
0, & \text{otherwise}.
\end{cases}
\end{align*}
\end{lem}

\begin{proof}
Denote by $\TwiParFunWei{ij}{\alpha}$ the right-hand side of (\ref{eq: cimasoni}).
We can write 
\begin{align*}
\der{\alpha}\TwiParFunWei{ij}{\alpha} & =\frac{1}{\CriPar}\sum_{\xi\in\EveSub{\torus}}\abs{\xi}\CriPar^{\abs{\xi}}\loopfactor{ij}=\frac{1}{\CriPar}\sum_{e}\sum_{\xi\in\EveSub{\torus}}\indicator_{e\in\xi}\CriPar^{\abs{\xi}}\loopfactor{ij}\\
 & =\frac{\abs{\torus}}{\CriPar}\left(\sum_{\xi\in\EveSub{\torus}}\indicator_{e_{H}\in\xi}\CriPar^{\abs{\xi}}\loopfactor{ij}+\sum_{\xi\in\EveSub{\torus}}\indicator_{e_{V}\in\xi}\CriPar^{\abs{\xi}}\loopfactor{ij}\right),
\end{align*}
where we used that $\torus$ has exactly $|\torus|$ horizontal edges
and $|\torus|$ vertical edges, and $\sum_{\xi\in\EveSub{\torus}}\indicator_{e\in\xi}\CriPar^{\abs{\xi}}\loopfactor{ij}$
only depends on whether $e$ is vertical or horizontal. Now, using
that $\indicator_{e\in\xi}=\frac{1}{\alpha+\alpha^{-1}}(\alpha^{-1}\indicator_{e\in\xi}-\alpha\indicator_{e\notin\xi})+\frac{\alpha}{\alpha+\alpha^{-1}}$,
we arrive at 
\begin{equation}
\der{\alpha}\TwiParFunWei{ij}{\alpha}=\frac{2|\torus|}{\alpha+\alpha^{-1}}\TwiParFunWei{ij}{\alpha}+\frac{|\torus|\alpha^{-1}}{\alpha+\alpha^{-1}}\left(\DisSpiObsVal{ij}{e_{H}}+\DisSpiObsVal{ij}{e_{V}}\right)\label{eq: B_Z_dZ}
\end{equation}

We now use Theorem \ref{LEMMACIMASONILEMMA} and Theorem \ref{thm: det_KW}
to compute $\der{\alpha}\TwiParFunWei{ij}{\alpha}$ at $\alpha=\alpha_{c}$, starting with the case $(ij)\neq(00)$.
It is straightforward to see from (\ref{eq: def_v_alpha_q}) that
\begin{align}
\TwiKacWarWeiDetFac{\CriPar_{c}}q & =2\CriPar_{c}^{2}\left(4-z\left(q\right)-z\left(q\right)^{-1}-w\left(q\right)-w\left(q\right)^{-1}\right);\label{eq: v_crit}\\
\der{\alpha}_{\alpha=\CriPar_{c}}\TwiKacWarWeiDetFac{\alpha}q & =2\sqrt{2}\CriPar_{c}^{2}\left(4-z\left(q\right)-z\left(q\right)^{-1}-w\left(q\right)-w\left(q\right)^{-1}\right).\label{eq: v_crit_prime}
\end{align}
In particular, since $|z(q)|=|w(q)|=1$, the quantity $\TwiKacWarWeiDetFac{\CriPar_{c}}q$
can only vanish when $z(q)=w(q)=1$, that is, $q\in\Z^{2}$. If $i=1$
(respectively, $j=1$), then, for $q\in\Lambda^{\star}+\TwiKacWarForShiVec{ij}$,
we have $\re{q}\re{\omega^\delta_{1}}+\im{q}\im{\omega^\delta_{1}}\in\Z+\frac{1}{2}$
(respectively, $\re{q}\re{\omega^\delta_{2}}+\im{q}\im{\omega^\delta_{2}}\in\Z+\frac{1}{2}$).
Since $\omega^\delta_{1,2}\in\Z^{2}$, this cannot happen for $q\in\Z^{2}$.
We conclude that if $(i,j)\neq(0,0)$, then $\TwiParFunWei{ij}{\alpha}\neq0$
. In this case, we have 
\begin{align*}
\der{\alpha}_{\alpha=\CriPar_{c}}\log\TwiParFunWei{ij}{\alpha}=\der{\alpha}_{\alpha=\CriPar_{c}}\frac{1}{2}\log\det\KW^{ij}=\frac{1}{2}\sum_{q\in\sfrac{\Z^{2}}{\TorLat^{*}}+\TwiKacWarForShiVec{ij}}\der{\alpha}_{\alpha=\CriPar_{c}}\log\TwiKacWarWeiDetFac{\CriPar}q.
\end{align*}
By (\ref{eq: v_crit}–\ref{eq: v_crit_prime}), each term in the sum
equals $\sqrt{2}$, and $|\sfrac{\Z^{2}}{\TorLat^{*}}+\TwiKacWarForShiVec{ij}|=|\torus|$.Thus
\begin{align*}
\der{\alpha}_{\alpha=\CriPar_{c}}\TwiParFunWei{ij}{\alpha}=\frac{\sqrt{2}|\torus|}{2}\TwiParFunWei{ij}{\alpha_{c}}.
\end{align*}
Since $2/(\alpha_{c}+\alpha_{c}^{-1})=\frac{\sqrt{2}}{2}$, plugging
this equation into (\ref{eq: B_Z_dZ}) at $\alpha=\alpha_{c}$ yields
the desired result $\DisSpiObsVal{ij}{\HorEdg}+\DisSpiObsVal{ij}{\VerEdg}=0$ if $(ij)\neq(00)$.

We now turn to the the case $(ij)=(00)$. We have $0\in\sfrac{\Z^{2}}{\TorLat^{*}}+\TwiKacWarForShiVec{00}=\sfrac{\Z^{2}}{\TorLat^{*}}$
and $\TwiKacWarWeiDetFac{\CriPar_{c}}0=0$; by (\ref{eq: v_crit_prime}),
we also have $\left.\der{\alpha}\right|_{\alpha=\CriPar_{c}}\TwiKacWarWeiDetFac{\CriPar}0=0$,
and from (\ref{eq: def_v_alpha_q}), it is straightforward to compute
\begin{align*}
\left.\frac{\mathrm{d}^{2}}{\mathrm{d}\alpha^{2}}\right|_{\alpha=\CriPar_{c}}\TwiKacWarWeiDetFac{\CriPar}0=16.
\end{align*}
By the above discussion, $\TwiKacWarWeiDetFac{\CriPar_{c}}q\neq0$
for other $q\in\sfrac{\Z^{2}}{\TorLat^{*}}$. Therefore, 
\begin{align*}
\TwiParFunWei{00}{\alpha}=\left(8(\alpha-\alpha_{c})^{2}\prod_{q\in\sfrac{\Z^{2}}{\TorLat^{*}}\setminus\{0\}}\TwiKacWarWeiDetFac{\CriPar_{c}}q\right)^{\frac{1}{2}}(1+\oo 1)\quad\text{as }\alpha\to\alpha_{c},
\end{align*}
and we obtain, taking into account Proposition \ref{prop: KW_to_Laplacian},
\begin{align}
\left.\der{\alpha}\right|_{\alpha=\CriPar_{c}}\TwiParFunWei{00}{\alpha}=-2\sqrt{2}\left(\prod_{q\in\sfrac{\Z^{2}}{\TorLat^{*}}\setminus\{0\}}\TwiKacWarWeiDetFac{\CriPar_{c}}q\right)^{\frac{1}{2}}=-2^{\frac{|\torus|+2}{2}}\alpha_{c}^{|\torus|-1}\sqrt{\text{det}^{\star}\Delta^{00}}.\label{eq: last_identity}
\end{align}
The choice of the \emph{negative} sign of the square root can be justified
by looking into \cite[Proof of Theorem 1.1]{CIMASONIDUMINIL-COPIN}.
For in that proof it is shown that $\det\KW_{00}(\alpha)=Z_{00}(\alpha)$
(Lemma 3.1) behaves as follows: when $\alpha\to0,$ $\Z_{00}(\alpha)\to1$,
and when $\alpha\to1$, $\Z_{00}(\alpha)\to-2|\T|$, and $\Z_{00}(\alpha)=0$
has a unique solution, that is, the critical point: $\Z_{00}(\alpha_{c})=0$.
Thus its derivative at the critical point must be negative.

Since $\TwiParFunWei{00}{\alpha_c}=0$, plugging (\ref{eq: last_identity})
into (\ref{eq: B_Z_dZ}) proves the desired result. 
\end{proof}
\begin{proof}[Proof of Theorem \ref{thm: sum}]
 By summing (\ref{eq: energy_B}) over a horizontal and vertical
edge and taking into account Lemma \ref{lem: B_ij}, we get 
\begin{align}
\SumEneDen=-\frac{1}{2\sqrt{2}}\frac{1}{\ParFunIsiRed}\left(\DisSpiObsVal{00}{\HorEdg}+\DisSpiObsVal{00}{\VerEdg}\right)=2^{\frac{|\T^{\delta}|+2}{2}}\alpha_{c}^{|\T^{\delta}|}\frac{1}{\ParFunIsiRed}\sqrt{\text{det}^{\star}\Delta^{00}}\frac{1}{\left|\T^{\delta}\right|}.\label{eq: proof_thm_2}
\end{align}
Recall that $\TwiParFunWei{00}{\alpha_{c}}=0$. By combining Theorem
\ref{LEMMACIMASONILEMMA}, the equation (\ref{eq: cimasoni}) and
Proposition \ref{prop: KW_to_Laplacian}, we get 
\begin{align*}
\ParFunIsiRed & =\frac{1}{2}\left(\TwiParFunWei{01}{\alpha_{c}}+\TwiParFunWei{10}{\alpha_{c}}+\TwiParFunWei{11}{\alpha_{c}}\right)\\
 & =\frac{1}{2}\left(\sqrt{\det\KW^{01}}+\sqrt{\det\KW^{10}}+\sqrt{\det\KW^{11}}\right)\\
 & =2^{\frac{|\T^{\delta}|}{2}-1}\alpha_{c}^{|\T^{\delta}|}\left(\sqrt{\det\Delta^{01}}+\sqrt{\det\Delta^{10}}+\sqrt{\det\Delta^{11}}\right).
\end{align*}

Plugging this into (\ref{eq: proof_thm_2}) ends the proof of the
Theorem \ref{thm: sum}.
\end{proof}
We now turn to the proof of Corollary \ref{THEOREM1}. We scale our
lattice and the torus by $\delta$, as in the introduction. Given
a lattice $\Lambda=n\omega_{1}+m\omega_{2}\subset\R^{2}$, we consider
the continuous Laplacian $\Delta=-\frac{\partial^{2}}{\partial x^{2}}-\frac{\partial^{2}}{\partial y^{2}}$
acting on the of twice continuously differentiable functions on the
torus $\R^{2}/\Lambda$. We then define the (Minakshisundaram–Pleijel)
zeta function by 
\begin{align*}
\zeta_{\T}(s):=\sum_{\lambda_{n}\neq0}\lambda_{n}^{-s},\quad\re{s}>1
\end{align*}
where the sum is over all non-zero eigenvalues of $\Delta$. This
$\zeta$-functions can be analytically continued into the vicinity
of $0$, and the $\zeta$-regularized determinant of the Laplacian
is defined by 
\begin{align*}
\log\text{det}_{\zeta}^{\star}\Delta & =-(\zeta_{\T})'(0).
\end{align*}

We will derive Corollary \ref{THEOREM1} from the following result
relating asymptotics of determinant of discrete Laplacian and to $\text{det}_{\zeta}^{\star}\Delta$:
\begin{prop}
(\cite[Theorem 1]{CHINTAJORGENSONKARLSSON1}, see also \cite{izyurov2020asymptotics})\label{PROPOSITIONASYMPTOTICSOFLAPLACIAN}
As $\delta\to0$ and $\omega_{1,2}^{\delta}\to\omega_{1,2}$, one
has 
\begin{align}
\mathrm{det}^{\star}\Delta_{\delta}^{00} & =\delta^{-2}\exp\left(C|\torus|)\right)\mathrm{det}_{\zeta}^{\star}\Delta\cdot(1+\oo 1),\label{prop: detlap_2}
\end{align}
where $C$ is an explicit constant. 
\end{prop}

Another ingredient is the classical computation of $\text{det}_{\zeta}^{\star}\Delta$,
due to Kronecker: 
\begin{align}
\text{det}_{\zeta}^{\star}\Delta=\im\tau|\T||\eta(\tau)|^{4},\label{eq: DetLapKronecker}
\end{align}
where $\eta(\tau)=\left(\frac{1}{2}\theta_{2}(\tau)\theta_{3}(\tau)\theta_{4}(\tau)\right)^{\frac{1}{3}}$
is the Dedekind eta-function, see e. g. \cite[section 10.2]{DIFRANCESCOMATHIEUSENECHAL}.
\begin{proof}[Proof of Corollary \ref{THEOREM1}]

In this proof, we will use several lattices, and hence we will use
the notation $V_{\omega_{1}^{\delta},\omega_{2}^{\delta}}^{ij}$ and
$\Delta_{\omega_{1}^{\delta},\omega_{2}^{\delta}}^{ij}$ for $V_{\delta}^{ij}$
and $\Delta_{\delta}^{ij}$, respectively, emphasizing the dependence
on the periods. Note that 
\begin{align*}
V_{2\omega_{1}^{\delta},\omega_{2}^{\delta}}^{00}=V_{\omega_{1}^{\delta},\omega_{2}^{\delta}}^{00}\oplus V_{\omega_{1}^{\delta},\omega_{2}^{\delta}}^{10};\quad & V_{\omega_{1}^{\delta},2\omega_{2}^{\delta}}^{00}=V_{\omega_{1}^{\delta},\omega_{2}^{\delta}}^{00}\oplus V_{\omega_{1}^{\delta},\omega_{2}^{\delta}}^{01};\quad & V_{2\omega_{1}^{\delta},\omega_{2}^{\delta}}^{01}=V_{\omega_{1}^{\delta},\omega_{2}^{\delta}}^{01}\oplus V_{\omega_{1}^{\delta},\omega_{2}^{\delta}}^{11},
\end{align*}
and the direct summands are invariant subspaces for the discrete Laplacian.
Hence, using the asymptotics (\ref{prop: detlap_2}), we obtain 
\begin{align}
\det\Delta_{\omega_{1}^{\delta},\omega_{2}^{\delta}}^{10}=\frac{\det^{\star}\Delta_{2\omega_{1}^{\delta},\omega_{2}^{\delta}}^{00}}{\det^{\star}\Delta_{\omega_{1}^{\delta},\omega_{2}^{\delta}}^{00}} & =e^{C|\torus|+\oo 1}\frac{\det_{\zeta}^{\star}\Delta_{2\omega_{1},\omega_{2}}}{\det_{\zeta}^{\star}\Delta_{\omega_{1},\omega_{2}}};\label{eq: det_1_reduce}\\
\det\Delta_{\omega_{1}^{\delta},\omega_{2}^{\delta}}^{01}=\frac{\det^{\star}\Delta_{\omega_{1}^{\delta},2\omega_{2}^{\delta}}^{00}}{\det^{\star}\Delta_{\omega_{1}^{\delta},\omega_{2}^{\delta}}^{00}} & =e^{C|\torus|+\oo 1}\frac{\det_{\zeta}^{\star}\Delta_{\omega_{1},2\omega_{2}}}{\det_{\zeta}^{\star}\Delta_{\omega_{1},\omega_{2}}};\label{eq: det_2_reduce}\\
\det\Delta_{\omega_{1}^{\delta},\omega_{2}^{\delta}}^{11}=\frac{\det\Delta_{2\omega_{1}^{\delta},\omega_{2}^{\delta}}^{01}}{\det^{\star}\Delta_{\omega_{1}^{\delta},\omega_{2}^{\delta}}^{01}} & =\frac{\det^{\star}\Delta_{2\omega_{1}^{\delta},2\omega_{2}^{\delta}}^{00}\det^{\star}\Delta_{\omega_{1}^{\delta},\omega_{2}^{\delta}}^{00}}{\det^{\star}\Delta_{2\omega_{1}^{\delta},\omega_{2}^{\delta}}^{00}\det^{\star}\Delta_{\omega_{1}^{\delta},2\omega_{2}^{\delta}}^{00}}\label{eq: det_3_reduce}\\
 & =e^{C|\torus|+\oo 1}\frac{\det_{\zeta}^{\star}\Delta_{2\omega_{1},2\omega_{2}}\det_{\zeta}^{\star}\Delta_{\omega_{1},\omega_{2}}}{\det_{\zeta}^{\star}\Delta_{2\omega_{1},\omega_{2}}\det_{\zeta}^{\star}\Delta_{\omega_{1},2\omega_{2}}}\nonumber 
\end{align}

By plugging $z=0$ into the reduction identities for theta functions
\cite[20.7.11-20.7.12]{NIST:DLMF}, we get 
\begin{align*}
\theta_{2}(2\tau)\theta_{3}(2\tau)\theta_{4}(2\tau) & =\frac{1}{2}\theta_{2}^{2}(\tau)\theta_{3}^{\frac{1}{2}}(\tau)\theta_{4}^{\frac{1}{2}}(\tau)\\
\theta_{2}\left(\tau/2\right)\theta_{3}\left(\tau/2\right)\theta_{4}\left(\tau/2\right) & =\sqrt{2}\theta_{2}^{\frac{1}{2}}(\tau)\theta_{3}^{\frac{1}{2}}(\tau)\theta_{4}^{2}(\tau).
\end{align*}
Using these identities and (\ref{eq: DetLapKronecker}), we get 
\begin{align*}
\frac{\det_{\zeta}^{\star}\Delta_{,2\omega_{1},\omega_{2}}}{\det_{\zeta}^{\star}\Delta_{\omega_{1},\omega_{2}}}=\frac{|\eta(\tau/2)|^{4}}{|\eta(\tau)|^{4}}=\frac{|\theta_{4}(\tau)|^{2}}{|\eta(\tau)|^{2}};\\
\frac{\det_{\zeta}^{\star}\Delta_{\omega_{1},2\omega_{2}}}{\det_{\zeta}^{\star}\Delta_{\omega_{1},\omega_{2}}}=4\frac{|\eta(2\tau)|^{4}}{|\eta(\tau)|^{4}}=\frac{|\theta_{2}(\tau)|^{2}}{|\eta(\tau)|^{2}};\\
\frac{\det_{\zeta}^{\star}\Delta_{2\omega_{1},2\omega_{2}}\det_{\zeta}^{\star}\Delta_{\omega_{1},\omega_{2}}}{\det_{\zeta}^{\star}\Delta_{2\omega_{1},\omega_{2}}\det_{\zeta}^{\star}\Delta_{\omega_{1},2\omega_{2}}}=\frac{|\eta(\tau)|^{8}}{|\eta(\tau/2)|^{4}|\eta(2\tau)|^{4}}=\frac{|\theta_{3}(\tau)|^{2}}{|\eta(\tau)|^{2}}
\end{align*}
Combining this with (\ref{eq: det_1_reduce})–(\ref{eq: det_3_reduce})
and and plugging these asymptotics and (\ref{eq: DetLapKronecker})
into the expression obtained in Theorem \ref{thm: difference}, we
get 
\begin{align*}
\SumEneDen=4\frac{(\im\tau)^{\frac{1}{2}}|\T|^{\frac{1}{2}}|\eta(\tau)|^{3}\delta^{-1}}{|\theta_{2}(\tau)|+|\theta_{3}(\tau)|+|\theta_{4}(\tau)|}\frac{1}{|\torus|}\left(1+\oo 1\right)\\
=\frac{4(\im\tau)^{\frac{1}{2}}|\eta(\tau)|^{3}}{|\T|^{\frac{1}{2}}\left(|\theta_{2}(\tau)|+|\theta_{3}(\tau)|+|\theta_{4}(\tau)|\right)}\delta+\oo{\delta},
\end{align*}
where we have noticed that $|\torus|\delta^{2}\to|\T|$. This concludes
the proof. 
\end{proof}

\section{Discrete holomorphic fermionic observables on a torus}

\label{sec: disc_hol}

To prove Theorems \ref{thm: difference} and \ref{thm: multipoint},
we employ the discrete holomorphic fermionic observables. We follow
the definitions and constructions in \cite{CHI_Mixed}; note that
the lattice there is scaled by $\sqrt{2}$ and rotated by $\pi/4$.
By a \emph{corner} of the lattice $\delta\Z^{2}$, we mean a midpoint
of a segment joining a lattice vertex $v\in\delta\Z^{2}$ with a vertex
of the dual lattice $(\delta\Z^{2})^{\star}=\delta\Z^{2}+\frac{\delta}{2}+\i\frac{\delta}{2}$.
The corners thus form another square lattice $\Cgr_{\delta}:=\frac{\delta}{2}\Z^{2}+\frac{\delta}{4}+\i\frac{\delta}{4}$.
We also define the corner graph $\Cgr(\torus)$ and the dual
graph $\Tdual$ of the torus $\torus$ in an obvious way. Given a
corner $z\in\Cgr(\torus)$, we denote by $z^{\circ}$ (respectively,
$z^{\bullet}$) the vertex of $\torus$ (respectively, $\Tdual$)
incident to $z$.

By the \emph{doubling} $\Tdbl^{\delta}$ of the torus $\torus$, we
mean the torus 
\begin{align*}
\Tdbl^{\delta}:=\delta\Z^{2}/\{2m\omega_{1}^{\delta}+2n\omega_{2}^{\delta}:m,n\in\Z^{2}\}.
\end{align*}
Note that $\Tdbl^{\delta}$ naturally forms a 4-sheet cover of $\torus$.
Given a corner $a\in\Cgr(\Tdbl^{\delta})$, we denote 
\begin{align*}
a_{pq}:=a+p\omega_{1}^{\delta}+q\omega_{2}^{\delta},\quad p,q\in\{0,1\},
\end{align*}
the four point that project to the same image in $\torus$.

Given a subset $\gamma$ of edges of $\Tdual$, we define the \emph{disorder
observable} 
\begin{align*}
\mu_{\gamma}:=e^{-2\beta\sum_{(xy)\cap\gamma\neq0}\sigma_{x}\sigma_{y}},
\end{align*}
where the sum is over the nearest neighbors $x\sim y\in\torus$ such
that the edge $(xy)$ intersects $\gamma$. We view subsets of edges
of $\Tdual$ as chains modulo $2$, in particular, $\gamma_{1}+\gamma_{2}=\gamma_{1}-\gamma_{2}$
will denote the symmetric difference of $\gamma_{1}$ and $\gamma_{2}$,
and $\partial\gamma$ is boundary of $\gamma$, that is, the set of
all vertices of $\Tdual$ incident to an odd number of edges in $\gamma$.
It turns out that the correlation of $\mu_{\gamma}$ with spins only
depends on $\partial\gamma$ and a homological class of $\gamma$
modulo $2$, as we now describe:
\begin{lem}
Let $v_{1},\dots,v_{2n}$ be vertices of $\torus$, and $\gamma_{1,2}$
two subsets of edges of $\Tdual$ with $\partial\gamma_{1}=\partial\gamma_{2}$.
Assume that $[\gamma_{1}]-[\gamma_{2}]=0$ in $H_{1}(\Tdual,\Z_{2})$, so that there exists a collection $F$ of faces of $\Tdual$ such that $\sum_{v\in F}\partial v=\gamma$. Then
\begin{align*}
\E[\sigma_{v_{1}}\dots\sigma_{v_{2n}}\mu_{\gamma_{1}}]=(-1)^N\E[\sigma_{v_{1}}\dots\sigma_{v_{2n}}\mu_{\gamma_{2}}],
\end{align*}
where $N=|\{v_{1},\dots,v_{2n}\}\cap F|$.
\end{lem}

\begin{proof}
Using $\gamma_{1}$ as a branch cut, one can construct a double cover $\tilde{\T}_{1}$ of the graph $\torus$. Namely, consider
two copies (\textquotedbl{}sheets\textquotedbl{}) of $\torus$, remove
in each copy the edges crossing $\gamma_{1}$, and add instead the
two edges connecting the corresponding vertices on different sheets.
We can write 
\begin{align*}
\E[\sigma_{v_{1}}\dots\sigma_{v_{2n}}\mu_{\gamma_{1}}] & =\frac{1}{Z}\sum_{\sigma:\torus\to\{\pm1\}}\sigma_{v_{1}}\dots\sigma_{v_{2n}}e^{\beta\sum_{x\sim y}\sigma_{x}\sigma_{y}}e^{-2\beta\sum_{(xy)\cap\gamma\neq0}\sigma_{x}\sigma_{y}}\\
 & =\frac{1}{Z}\sum_{\substack{\sigma:\tilde{\T}_{1}\to\{\pm1\}\\
\sigma_{v}\equiv-\sigma_{v^{*}}
}
}\sigma_{v_{1}}\dots\sigma_{v_{2n}}e^{\frac{\beta}{2}\sum_{x\sim y}\sigma_{x}\sigma_{y}},
\end{align*}
where $v\in\torus$ is identified with its copy on the first sheet,
and $v^{*}$ denote its copy on the second sheet. We, if we repeat the same construction with $\gamma_2$, since $\gamma_{1}-\gamma_{2}=\sum_{v\in F}\partial v$,
then flipping the sheets of all vertices corresponding to $v\in F$
gives an isomorphism of the two double covers, under which exactly
$N$ of $v_{1},\dots,v_{2n}$ will move to the second sheet. This
gives the desired result. 
\end{proof}
In view of this lemma, the quantity $\E[\sigma_{v_{1}}\dots\sigma_{v_{2n}}\mu_{\gamma}]$
can be understood as a \textquotedbl{}multi-valued function\textquotedbl{}
of $v_{1},\dots,v_{2n}$ and $u_{1},\dots,u_{2m}$, where $\partial\gamma=\{u_{1},\dots,u_{2m}\}$,
meaning that once the initial condition (i. e., the choice of $\gamma$
for a certain position of $v_{1},\dots,u_{2m}$) is prescribed, there
is a natural way to extend its value as the marked points move around
in in the lattice. Namely, when moving $u_{i}\mapsto u'_{i}\sim u_{i}$,
one replaces $\gamma$ by $\gamma+(u_{i}u_{i}')$, and when moving
$v_{i}\mapsto v'_{i}$, one adds a $-$ sign whenever $(v_{i},v'_{i})$
crosses $\gamma$. With this convention, if all the points move and
come back to their initial positions, say, without leaving a fixed
fundamental domain of $\torus$, then the expression $\E[\sigma_{v_{1}}\dots\sigma_{v_{2n}}\mu_{\gamma}]$
changes sign in the same way as $\prod(u_{i}-v_{j})^{\frac{1}{2}}$

For $p,q\in\{0,1\}$, we denote 
\begin{align}
\mu_{pq}:=\mu_{\gamma_{pq}},
\end{align}
where $\gamma_{pq}$ is a simple loop on $\Tdual$ that lifts to a
path on $\delta\Z^{2}$ connecting a point $z$ with $z+p\omega_{1}^{\delta}+q\omega_{2}^{\delta}$.
In particular, we can choose $\mu_{00}=\emptyset$.
\begin{defn}
We define the \emph{Dirac spinor}
\begin{align}
\eta_{z}:=e^{\frac{\i\pi}{4}}\left(\frac{z^{\bullet}-z^{\circ}}{|z^{\bullet}-z^{\circ}|}\right)^{-\frac{1}{2}},\label{eq: Dirac}
\end{align}
understood as a two-valued function on the corner lattice $\Cgr(\torus)$,
i.e., the function the double cover of that graph ramified at every
face, see Fig. \ref{Fig: Dirac}. 
\end{defn}

\begin{figure}
\includegraphics[width=0.8\textwidth]{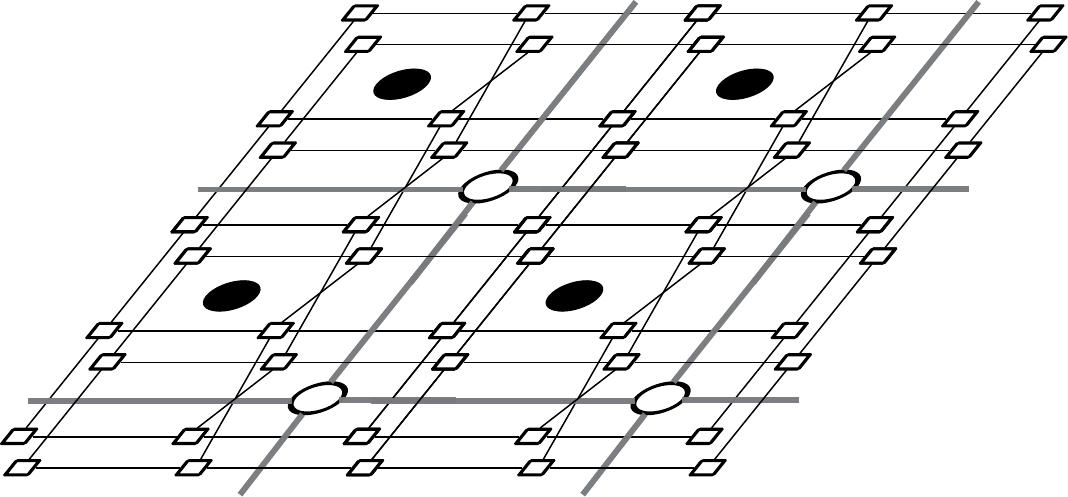}\caption{A piece of the double cover of the corner graph, ramified at every face of that graph, on which the Dirac spinor $\eta_z$ is defined. The white discs are the vertices of the original lattice $\delta\Z^2$, the black discs are vertices of the dual graph, and the small squares are vertices of the corner graph.} \label{Fig: Dirac}
\end{figure}

The fermionic observables we use are similar to the to multi-point
fermionic observables as in \cite[Section 2.4]{CHI_Mixed}. Since
we will only need those for the energy correlation, we will restrict
our consideration to a very specific configuration.

Let $\Cgr^{[e_{1}\dots e_{k}]}(\hat{\T}^{\delta})$ denote the graph
obtained by deleting from $\Cgr(\hat{\T}^{\delta})$ the edges crossing
$e_{1},\dots,e_{k}$ and their shifts by by $\omega_{1}^{\delta},\omega_{2}^{\delta},\omega_{1}^{\delta}+\omega_{2}^{\delta}$.
The quantity $F_{e_{1}\dots e_{k}}(a,\cdot)$ defined below in (\ref{eq: def_obs})
is a ``spinor'' on $\Cgr(\hat{\T}^{\delta})$ ramified at $a^{\bullet},a^{\circ}$
and their shifts by $\omega_{1}^{\delta},\omega_{2}^{\delta},\omega_{1}^{\delta}+\omega_{2}^{\delta}$,
i. e., a function on the double cover of $\Cgr(\hat{\T}^{\delta})$
ramified at those points and changing sign between sheets. We prefer
to view it as a function rather than spinor, by introducing a
branch cut $[a^{\circ}a^{\bullet}]$ that divides $a$ into two
vertices $a^{\pm}$, $a^{+}$ being on the left as seen from $a^{\bullet}$,
and similarly for the shifts. We denote the resulting graph by $\Cgr_{[a]}^{[e_{1}\dots e_{k}]}(\hat{\T}^{\delta}),$ see Fig. \ref{Fig: cut_graph}. 

\begin{figure}
\includegraphics[width=0.4\textwidth]{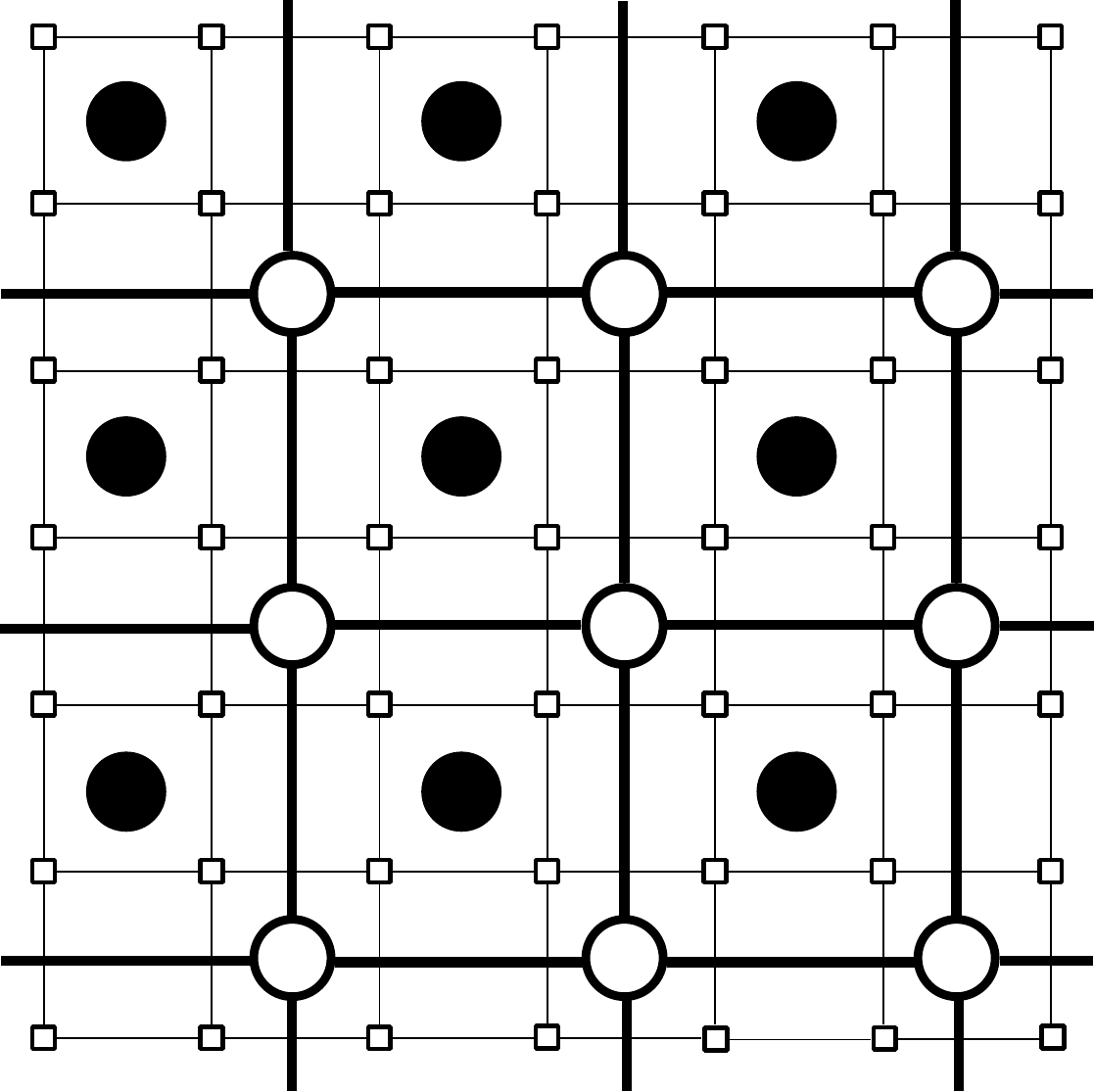}\qquad{}\includegraphics[width=0.4\textwidth]{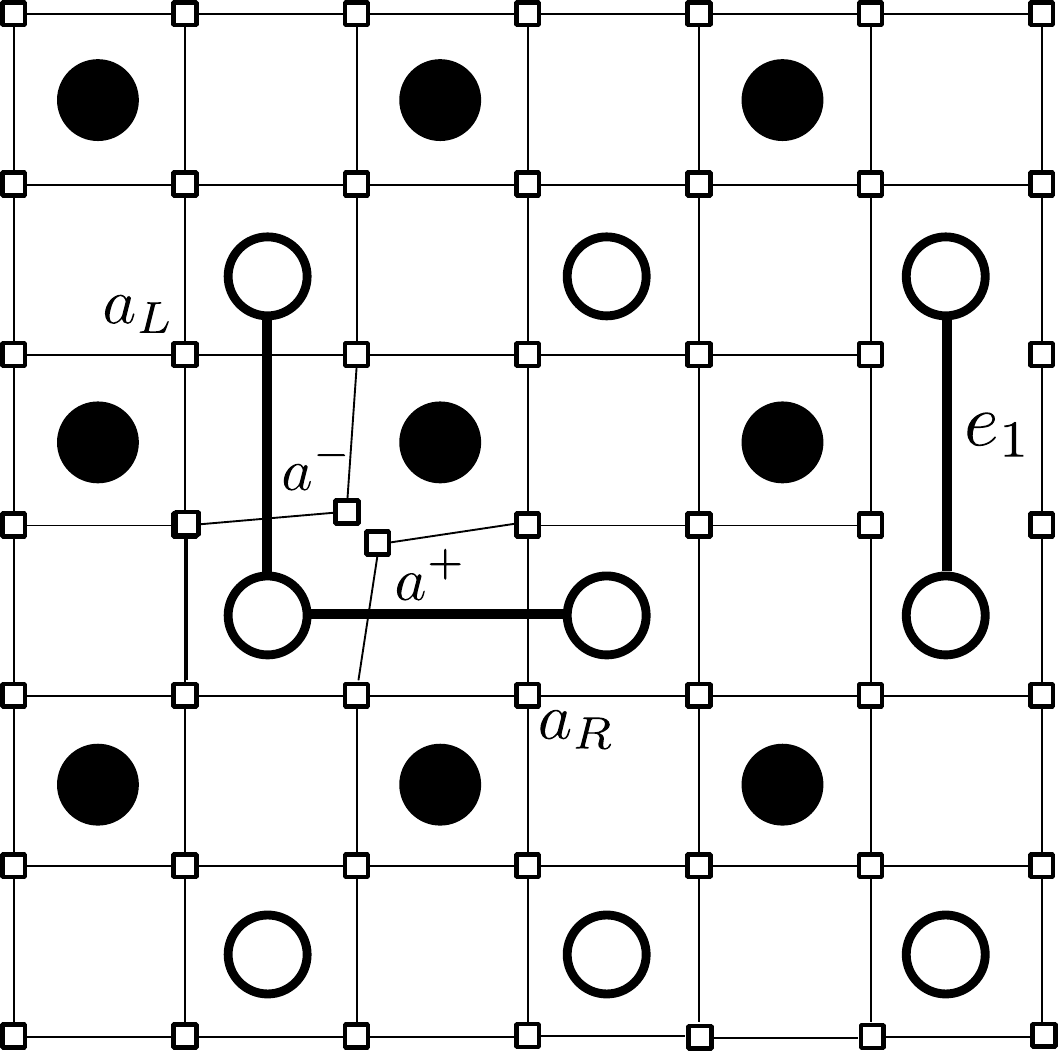}\caption{The lattice $\delta\protect\Z^{2}$ as $\circ$, its dual lattice
as $\bullet$, the corner lattice as $\square$. On the right: the
modified graph $\protect\Cgr_{[a]}^{[e_{1}]}(\cdot)$. Any closed
path on $\protect\Cgr_{[a]}(\cdot)$ either encircles both $a^{\bullet}$
and $a^{\circ}$, or none; hence $F(a,\cdot)$ is single-valued on
the modified graph. The values $F_{e_{1}}(a,a_{L})$ and $F_{e_{1}}(a,a_{R})$
are proportional to $\protect\E(\protect\en_{e_{L}}\protect\en_{e_{1}})$
and $\protect\E(\protect\en_{e_{R}}\protect\en_{e_{1}})$ respectively,
see Lemma \ref{lem: s-hol_spec_values}. }\label{Fig: cut_graph}
\end{figure}

\begin{defn}
Given distinct edges $e_{1},\dots,e_{k}$ in the torus $\T^{\delta}$
that pairwise have no common incident corner, and a corner $a\in\Cgr(\T^{\delta}),$
we define the observable 
\begin{equation}
F_{e_{1}\dots e_{k}}(a,z):=\eta_{z}\E\left[\sigma_{z^{\circ}}\sigma_{a^{\circ}}\mu_{z^{\bullet}}\mu_{a^{\bullet}}\prod_{i=1}^{k}\en_{e_{i}}\right], \quad z\in \Cgr_{[a]}^{[e_{1}\dots e_{k}]}(\hat{\T}^{\delta})\label{eq: def_obs}
\end{equation}
We fix a choice of a square root in the definition of $\eta_{a}$
and set the ``initial conditions'' for $F_{e_{1}\dots e_{k}}(a,z)$
of the above observable by $\mu_{a^{+,\bullet}a^{\bullet}}=\mu_{\emptyset}$
and $\eta_{a^{+}}=\eta_{a}$ , so that we have 
\[
F_{e_{1}\dots e_{k}}(a,a^{\pm})=\pm\eta_{a}\E\left[\prod_{i=1}^{k}\en_{e_{i}}\right].
\]
\end{defn}

Recall that a function $F$ defined on a corner graph is called \emph{s-holomorphic}
if it satisfies the phase condition $F(z)\in\eta_{z}\R$, and for
every edge $e$, one has 
\begin{align}
F(e_{NE})+F(e_{SW})=F(e_{NW})+F(e_{SE}),\label{eq: s-hol}
\end{align}
where $e_{NE},e_{SW},e_{NW},e_{SE}$ are the four corners neighboring
$e$. We will, in fact, only need the following simple properties
of s-holomorphic functions: if $F$ is s-holomorphic, then its restriction
to each of the two sub-lattices $\{z:\eta_{z}\in\eta_{z_{0}}\R\cup \i\eta_{z_{0}}\R\}$
are discrete holomorphic in the usual sense, that is, 
\begin{align}
F\left(z+\frac{\delta}{2}\right)-F\left(z-\frac{\delta}{2}\right)-\frac{1}{\i}\left(F\left(z+\i\frac{\delta}{2}\right)-F\left(z-\i\frac{\delta}{2}\right)\right)=0,\label{eq: discr_d_bar}
\end{align}
for every $z\in\Cgr$ such that $F$ is s-holomorphic at both edges
intersecting the segments $\left[z-\frac{\delta}{2};z+\frac{\delta}{2}\right]$
and $\left[z-\i\frac{\delta}{2};z+\i\frac{\delta}{2}\right]$. As a
consequence, if, for an open set $\Omega$, a sequence $F_{\delta}$
of s-holomorphic functions defined on refining lattices $\Omega\cap\Cgr(\delta\Z^{2})$
is uniformly bounded on every compact subset of $\Omega$, then there
is a holomorphic function $f:\Omega\to\C$ and a subsequence $\delta_{k}$
such that 
\begin{align*}
F_{\delta_{k}}(z)=\proj{\eta_{z}}{f(z)}+\oo{1},
\end{align*}
uniformly on compact subsets of $\Omega$. This is a simple consequence of the discrete Cauchy integral formula, see e.g. \cite{CHELKAKSMIRNOV}.
\begin{lem}
\label{lem: s-hol_spec_values}The observable $F_{e_{1}\dots e_{k}}(a,\cdot)$
is a well defined s-holomorphic function on $\Cgr_{[a]}^{[e_{1}\dots e_{k}]}(\Tdbl_{\delta}).$
Moreover, it has the following special values: 
\begin{align}
F_{e_{1}\dots e_{k}}(a,a^{\pm}) & =\pm\eta_{a}\E\left[\prod_{m=1}^{k}\en_{e_{m}}\right],\label{eq: F_value_sing_1}\\
F_{e_{1}\dots e_{k}}(a,a^{\pm}+p\omega_{1}^{\delta}+q\omega_{2}^{\delta}) & =\mp\eta_{a}\E\left[\mu_{pq}\prod_{m=1}^{k}\en_{e_{m}}\right], & (pq)\neq(00).\label{eq: F_value_sing_2}
\end{align}
If, moreover, we denote \textup{$\anw:=a+\i\cdot(a^{\bullet}-a^{\circ})$
}and $\ase=a-\i\cdot(a^{\bullet}-a^{\circ}),$ and by $e_{L,R}$
the edges incident to both $a$ and $a_{L,R}$ respectively, then

\begin{align}
F_{e_{1}\dots e_{k}}(a,a_{L,R}) & =-\i\eta_{a}\sqrt{2}\E\left[\en_{e_{L,R}}\prod_{m=1}^{k}\en_{e_{m}}\right],\label{eq: F_value_near}\\
F_{e_{1}\dots e_{k}}(a,a_{L,R}+p\omega_{1}^{\delta}+q\omega_{2}^{\delta}) & =\pm\i\eta_{a}\sqrt{2}\E\left[\mu_{pq}\en_{e_{L,R}}\prod_{m=1}^{k}\en_{e_{m}}\right], & (pq)\neq(00).\label{eq: F_value_near_2}
\end{align}
\end{lem}

\begin{proof}
First, observe that $F_{e_{1}\dots_{e_{k}}}(a,\cdot)$ is locally
well-defined, since when $z$ moves around any face of $\Cgr(\Tdbl)$
except one corresponding to $a_{ij}^{\bullet}$ or $a_{ij}^{\circ}$,
the $-1$ sign acquired by $\eta_{z}$ cancels the $-1$ sign coming
from winding of $z^{\bullet}$ around $z^{\circ}$. Moving $z$ from
$a^{+}$ to $a^{-}$ around, say, $a^{\circ}$ results only in the
sign change of $\eta_{z}$, hence, $F(a,a^{-})=-\eta_{a}$. Moving
$z$ \textquotedbl{}around the torus\textquotedbl{} to $z_{pq}=z+p\omega_{1}^{\delta}+q\omega_{2}^{\delta}$
results in replacing $\gamma$ defining $\mu_{a^{\bullet}}\mu_{z^{\bullet}}$
with $\gamma+\gamma_{pq}$; hence, doing so twice yields the same
value. Finally, a \emph{simple} loop lifting to a path connecting
$a$ with $a_{pq}$ must start with $a^{+}$ and end at $a_{pq}^{-}$,
hence the $-$ sign in (\ref{eq: F_value_sing_2}) when $(pq)\neq(00)$.
The proof of s-holomorphicity in \cite{CHI_Mixed} is completely local
and thus extend verbatim to our case. To prove (\ref{eq: F_value_near}),
we start with the value $F_{e_{1}\dots e_{k}}(a,a^{+})=\eta_{a}\E[\dots]$
and move in two steps to $z=a_{R}$ so that $z^{\bullet}-z^{\circ}$
rotates couterclockwise. This way, $z^{\circ}$ does not cross the newly
created disorder line $[a^{\bullet}a_{L}^{\bullet}]$ , and we get
\begin{multline*}
\eta_{a_{L}}\sigma_{a^{\circ}}\sigma_{\anw^{\circ}}\mu_{a^{\bullet}}\mu_{\anw^{\bullet}}=-\i\eta_{a}\sigma_{a^{\circ}}\sigma_{\anw^{\circ}}e^{-2\beta\sigma_{a^{\circ}}\sigma_{\anw^{\circ}}}\\
=-\i\eta_{a}\left(\sigma_{a^{\circ}}\sigma_{\anw^{\circ}}\cosh(-2\beta)+\sinh(-2\beta)\right)=-\sqrt{2}\i\eta_{a}\en_{\sigma_{a^{\circ}}\sigma_{\anw^{\circ}}}
\end{multline*}
inside the correlation, where we have used that $\cosh(-2\beta)=\sqrt{2}$
and $\sinh(-2\beta)=-1$. The computation for $a_{R}$ is identical,
except we start $F_{e_{1}\dots e_{k}}(a,a^{-})=-\eta_{a}\E[\dots]$
and move so that $z^{\bullet}-z^{\circ}$ rotates clockwise.
The proof of (\ref{eq: F_value_near_2}) is similar, taking into account
the disorder line $\gamma_{pq}$ created when moving from $a$ to
$a_{pq}$.
\end{proof}
It is convenient to anti-symmetrize $F_{e_{1}\dots e_{k}}$ by introducing
\[
F_{e_{1}\dots e_{k}}^{(ij)}(a,z):=\frac{1}{4}\sum_{p,q\in\{0,1\}}(-1)^{ip+jq}F(a,z+p\omega_{1}^{\delta}+q\omega_{2}^{\delta}).
\]

We summarize their properties in the following Lemma:
\begin{lem}
The observables $F_{e_{1}\dots e_{k}}^{(ij)}(a,\cdot)$ are s-holomorphic
functions on $\Cgr_{[a]}^{[e_{1}\dots e_{k}]}(\hat{\T}^{\delta})$
that satisfy the anti-periodicity condition

\begin{equation}
F_{e_{1}\dots e_{k}}^{(ij)}(a,z+p\omega_{1}^{\delta}+q\omega_{2}^{\delta})=(-1)^{ip+jq}F_{e_{1}\dots e_{k}}^{(ij)}(a,z).\label{eq: anti-symmetry}
\end{equation}
They have special values given by 
\begin{align}
F_{e_{1}\dots e_{k}}^{(ij)}(a,a^{\pm}) & =\pm\eta_{a}\E\left[\mu^{(ij)}\prod_{m=1}^{k}\en_{e_{m}}\right],\label{eq: F_ij_at_sing}\\
F_{e_{1}\dots e_{k}}^{(ij)}(a,a_{L,R}) & =-\i\eta_{a}\sqrt{2}\E\left[\mu^{(ij)}\en_{e_{L,R}}\prod_{m=1}^{k}\en_{e_{m}}\right]\label{eq: F_ij_at_LR}
\end{align}
where 
\begin{align}
\mu^{(ij)}:=\frac{1}{4}\sum_{pq\in\{0,1\}}(-1)^{(1-i)p+(1-j)q+pq}\mu_{pq}.\label{eq: def_mu^ij}
\end{align}
Moreover, there exists a constant $c\in\R$ (depending on $\T^{\delta}$ and on $a$ but not on $z$) such that $F^{(00)}(a,z)=\proj{\eta_{z}}{\i\eta_{a}c}.$
In particular, 
\begin{align}
F^{(00)}(a,\anw)=F^{(00)}(a,\ase).\label{eq: F^00 constant}
\end{align}
\end{lem}

\begin{proof}
The anti-periodicity (\ref{eq: anti-symmetry}) is manifest from the
construction, the s-holomorphicity follows from linearity, and (\ref{eq: F_ij_at_sing}–\ref{eq: F_ij_at_LR})
are obtained by simply summing (\ref{eq: F_value_sing_1}–\ref{eq: F_value_near_2}).
The fact that $F^{(00)}$ is essentially a constant is a discrete
analog of the claim that a meromorphic function on a torus with at
most one simple pole is constant. Indeed, consider the restriction
of $F^{(00)}(a,\cdot)$ onto the index 2 sub-lattice $\Cgr'(\torus):=\{z\in\Cgr(\torus):\eta_{z}\in e^{\i\pi/4}\eta_{a}\R\cup e^{-\i\pi/4}\eta_{a}\R\}$
of $\Cgr(\torus)$. As noticed above, this restriction is discrete
holomorphic, that is, the identity (\ref{eq: discr_d_bar}) holds
for every $z\in\Cgr(\torus)$ such that $\eta_{z}\in\eta_{a}\R\cup\i\eta_{a}\R$,
except, possibly, for $z=a$. However, summing (\ref{eq: discr_d_bar})
over \emph{all} $z$, we see that each value $F(z)$, $z\in\Cgr'(\torus)$,
enters the sum with the coefficient $1-1+\i-\i=0$. Therefore, (\ref{eq: discr_d_bar})
also holds for $z=a$. Since the restriction of a discrete holomorphic
function $F^{(00)}(a,\cdot)$ to each of the two sub-lattices $\{z\in\Cgr(\torus):\eta_{z}\in e^{\pm\i\pi/4}\eta_{a}\R\}$
is discrete harmonic, the maximum principle implies that these restrictions
are constant. The s-holomorphicity condition (\ref{eq: s-hol}) now
implies that the restrictions of $F^{(00)}(a,\cdot)$ to two other
sub-lattices $\{z\in\Cgr(\torus):\eta_{z}=\eta_{a}\R\}$ and
$\{z\in\Cgr(\torus):\eta_{z}=\i\eta_{a}\R\}$ are also constant
(and it's easy to see that the s-holomorphicity implies that in that
case, $F^{(00)}(a,z)=\proj{\eta_{z}}C$. Finally, since $F(a,a^{+})=-F(a,a^{-}),$
it is a simple direct check that the only value consistent with (\ref{eq: discr_d_bar})
vanishing at $a$ is $F(a,a^{+})=0,$ hence $F(a,z)\equiv0$ on the sublattice $\{z\in\Cgr(\torus):\eta_{z}=\pm\eta_{a}\},$
implying the restriction $C\in\i\eta_{a}\R.$
\end{proof}
We now describe the discrete analog of the Cauchy kernel $\frac{1}{z-a}$,
tailored for our purposes. Given $a\in\Cgr(\Z^{2})$ and a choice
of the square root at $\eta_{a}$ this is a unique s-holomorphic function
$\dzmone_{a}(\cdot)$ on $\Cgr_{[a]}(\Z^{2})$ such that 
\begin{align}
\dzmone_{a}(a^{+})=\eta_{a},\quad\dzmone_{a}(a^{-})=-\eta_{a},\label{eq: dzmone_res}
\end{align}
and, as $z\to\infty$, one has 
\begin{equation}
\dzmone_{a}(z)=\frac{\sqrt{2}}{\pi}\proj{\eta_{z}}{\bar{\eta}_{a}(z-a)^{-1}}+O(|z-a|^{-2}).\label{eq: asymp_zmone}
\end{equation}
Moreover, we have 
\begin{equation}
P_{a}(\anw)=P_{a}(\ase)=0.\label{eq: P_a_close}
\end{equation}
We also define the re-scaled version of this function, living on $\Cgr(\delta\Z^{2})$:
\begin{align*}
P_{a}^{\delta}(z)=\delta^{-1}\dzmone_{\delta^{-1}a}(\delta^{-1}z).
\end{align*}

The function $P_{a}$ is a multiple of the discrete Cauchy kernel
on the square lattice; see (\cite[Lemma 4.9]{CHI_Mixed}) for the
details. Since in (\cite{CHI_Mixed}), the lattice is scaled by $\sqrt{2}$
and rotated by $\frac{\pi}{4}$, one must re-define $P_{a}(z):=e^{\frac{\i\pi}{8}}P_{\sqrt{2}ae^{\i\frac{\pi}{4}}}\left(\sqrt{2}e^{\i\frac{\pi}{4}}z\right)$,
which is reflected in the pre-factor in (\ref{eq: asymp_zmone}).

The following lemma elucidates the singularity structure of $F_{e_{1}\dots e_{k}}^{(ij)}(a,z_{pq})$.
\begin{lem}
\label{lem: fmp_discrete_singularity}The function 
\begin{equation}
\tilde{F}_{a}(\cdot)=F_{e_{1}\dots e_{k}}^{(ij)}(a,\cdot)-\delta\cdot\E\left[\mu^{(ij)}\prod_{i=1}^{k}\en_{e_{i}}\right]\dzmone_{a}^{\delta}(\cdot)\label{eq: res_a}
\end{equation}
extends to an s-holomorphic function on $\Cgr(\T^{\delta})$ in a
neighborhood of $a$. If $z_{2m-1},z_{2m}$ are two corners
incident to $e_{m}$ and symmetric with respect to its centre, then 
\begin{multline}
\tilde{F}_{m}(\cdot)=F_{e_{1}\dots e_{k}}^{(ij)}(a,\cdot)\\
-\frac{\delta}{\sqrt{2}}\cdot\i\bar{\eta}_{z_{2m-1}}F_{e_{1}\dots\hat{e}_{m}\dots e_{k}}^{(ij)}(a,z_{2m})\dzmone_{z_{2m-1}}^{\delta}(\cdot)-\frac{\delta}{\sqrt{2}}\cdot\i\bar{\eta}_{z_{2m}}F_{e_{1}\dots\hat{e}_{m}\dots e_{k}}^{(ij)}(a,z_{2m-1})\dzmone_{z_{2m}}^{\delta}(\cdot)\label{eq: res_e_m}
\end{multline}
extends to an s-holomorphic function on $\Cgr(\T^{\delta})$ in a
neighborhood of $e_{m}$, where the choice of square root in $\eta_{z_{2m-1}}$
and $\dzmone_{z_{2m-1}}^{\delta}(\cdot)$ (resp. $\eta_{z_{2m}}$
and $\dzmone_{z_{2m}}^{\delta}(\cdot)$) are related by (\ref{eq: dzmone_res}).
Moreover, after that extension,
\begin{equation}
\tilde{F}(a)=\tilde{F}_{m}(z_{2m-1})=\tilde{F}_{m}(z_{2m})=0.\label{eq: values_at_e_i}
\end{equation}
\end{lem}

\begin{proof}
It follows from (\ref{eq: F_ij_at_sing}) and (\ref{eq: dzmone_res})
that $\tilde{F}_{a}(a^{\pm})=0$, and hence these two corners of $\Cgr_{[a]}^{[e_{1}\dots e_{k}]}(\torus)$
can be glued back together, yielding the first claim. 

For the second claim, we use that 
\begin{multline*}
\sigma_{z_{2m-1}^{\circ}}\sigma_{z_{2m}^{\circ}}\mu_{z_{2m-1}^{\bullet}}\mu_{z_{2m}^{\bullet}}=\sigma_{z_{2m-1}^{\circ}}\sigma_{z_{2m}^{\circ}}e^{-2\beta\sigma_{z_{2m-1}^{\circ}}\sigma_{z_{2m}^{\circ}}}\\
=\sigma_{z_{2m-1}^{\circ}}\sigma_{z_{2m}^{\circ}}\cosh(-2\beta)+\sinh(-2\beta)=\sqrt{2}\en_{e_{m}},
\end{multline*}
so that we can replace $\en_{m}$ by $2^{-\frac{1}{2}}\sigma_{z_{2m-1}^{\circ}}\sigma_{z_{2m}^{\circ}}\mu_{z_{2m-1}^{\bullet}}\mu_{z_{2m}^{\bullet}}$
in the definition of $F_{e_{1}\dots e_{k}}(a,\cdot).$ With this substitution,
it becomes a spinor on $\Cgr_{[a]}(\hat{\T}^{\delta})$ ramified at
$z_{2m-1}^{\circ},z_{2m-1}^{\bullet},z_{2m}^{\circ},z_{2m}^{\bullet},$
and their shifts by $\omega_{1}^{\delta},\omega_{2}^{\delta},\omega_{1}^{\delta}+\omega_{2}^{\delta}$.
We modify the graph by introducing branch cuts $[z_{2m-1}^{\circ},z_{2m-1}^{\bullet}]$
and $[z_{2m}^{\circ},z_{2m}^{\bullet}]$ (and their shifts) which
splits the vertices $z_{2m-1},z_{2m}$ into $z_{2m-1}^{\pm}$ and
$z_{2m}^{\pm}$, so that $F_{e_{1}\dots e_{k}}(a,\cdot)$ becomes
a function on the resulting graph obeying the property $F_{e_{1}\dots e_{k}}(a,z_{p}^{+})=-F_{e_{1}\dots e_{k}}(a,z_{p}^{-})$, $p=2m-1,2m$.
More concretely, 
\begin{multline*}
F_{e_{1}\dots e_{k}}(a,z_{2m-1}^{+})=\frac{1}{\sqrt{2}}\eta_{z_{2m-1}^{+}}\E\left[\sigma_{z_{2m}^{\circ}}\mu_{z_{2m}^{\bullet}}\sigma_{a^{\bullet}}\mu_{a^{\bullet}}\prod_{\substack{i=1\\
i\neq m
}
}^{k}\en_{e_{i}}\right]\\
=\frac{1}{\sqrt{2}}\eta_{z_{2m-1}^{+}}\bar{\eta}_{z_{2m}}F_{e_{1}\dots\hat{e}_{m}\dots e_{k}}^{(ij)}(a,z_{2m});
\end{multline*}
\begin{multline*}
F_{e_{1}\dots e_{k}}(a,z_{2m}^{+})=\frac{1}{\sqrt{2}}\eta_{z_{2m}^{+}}\E\left[\sigma_{z_{2m-1}^{\circ}}\mu_{z_{2m-1}^{\bullet}}\sigma_{a^{\bullet}}\mu_{a^{\bullet}}\prod_{\substack{i=1\\
i\neq m
}
}^{k}\en_{e_{i}}\right]\\
=\frac{1}{\sqrt{2}}\eta_{z_{2m-1}^{+}}\bar{\eta}_{z_{2m}}F_{e_{1}\dots\hat{e}_{m}\dots e_{k}}^{(ij)}(a,z_{2m-1}).
\end{multline*}
Here the signs are related in such a way that $\eta_{z_{2m-1}^{+}}\bar{\eta}_{z_{2m}}=\eta_{z_{2m-1}^{+}}\bar{\eta}_{z_{2m}}=\i$
if $z_{2m-1}^{+}$ and $z_{2m}^{+}$ are on the \emph{outer} side
from the edge $e_{m}.$ Therefore, taking into account that $\dzmone_{z_{2m-1}}^{\delta}(z_{2m})=0$
and $\dzmone_{z_{2m}}^{\delta}(z_{2m-1})=0,$ we see that $\tilde{F}_{m}(z_{2m-1}^{\pm})=\tilde{F}_{m}(z_{2m}^{\pm})=0$,
so that $z_{2m-1}^{+}$ and $z_{2m-1}^{-}$ can be glued back together,
and similarly for $z_{2m}^{\pm}.$ These results hold with the same
proof for shifts of $e_{m}$ by $\omega_{1}^{\delta},\omega_{2}^{\delta},\omega_{1}^{\delta}+\omega_{2}^{\delta}$
and hence extend to $F^{(ij)}(a,z^{\pm})$ by linearity, thus proving
the claim.
\end{proof}
\begin{rem}
For $(ij)\neq(00),$ assuming $\E\mu^{(ij)}$ are known, the formulae
(\ref{eq: res_a}–\ref{eq: res_e_m}) identify the functions $F_{e_{1}\dots e_{k}}^{(ij)}$
uniquely by recursion. They \emph{do not} identify $F_{e_{1}\dots e_{k}}^{(00)}$ uniquely,
because of a possibility of adding a ``constant'' s-holomorphic
function $\proj{\eta_{z}}c,$ $c\in\C.$ However, if $k\geq1$, (\ref{eq: values_at_e_i})
removes this degree of freedom. We were unable to identify this
constant and its asymptotics for $k=0$ based on discrete holomorphicity
considerations only, hence the input from Corollary \ref{THEOREM1}
is needed to start the induction. For the energy \emph{difference}
of Theorem \ref{thm: difference}, the value of this is constant does
not matter as it cancels out, since $\E\left[\mu^{(00)}\en_{V}\right]=\E\left[\mu^{(00)}\en_{H}\right]$
by (\ref{eq: F_ij_at_LR}, \ref{eq: F^00 constant}).
\end{rem}

We conclude this section by identifying $\E\mu^{(ij)}$ with Kac-Ward
determinants from Section \ref{sec: Kac-Ward}, and deducing their
asymptotics:
\begin{lem}
\label{lem: pf_ratios}We have
\begin{align*}
\E(\mu^{(ij)})=\frac{Z^{(ij)}}{Z^{(01)}+Z^{(10)}+Z^{(11)}},
\end{align*}
where $Z^{(ij)}$ is defined in (\ref{eq: cimasoni}). Therefore,
we have, as $\delta\to0,$ 
\begin{gather}
\E\mu^{(10)}\to\En{}{10}:=\frac{\PartFun{10}}{\PFTotal}=\frac{|\theta_{4}|}{|\theta_{2}|+|\theta_{3}|+|\theta_{4}|},\label{eq: En_1}\\
\E\mu^{(01)}\to\En{}{01}:=\frac{\PartFun{01}}{\PFTotal}=\frac{|\theta_{2}|}{|\theta_{2}|+|\theta_{3}|+|\theta_{4}|},\label{eq: En_2}\\
\E\mu^{(11)}\to\En{}{11}:=\frac{\PartFun{11}}{\PFTotal}=\frac{|\theta_{3}|}{|\theta_{2}|+|\theta_{3}|+|\theta_{4}|}.\label{eq: En_3}
\end{gather}
\end{lem}

\begin{proof}
We prove the identities using high-temperature expansion; since $\T^{\delta}$
is self-dual, we can freely pass between primal and dual lattice.
For $x,y\in(\T^{\delta})^{\star}$, $x\sim y$, recall the notation
(\ref{eq:def_phi}). The high-temperature expansion reads
\begin{multline*}
\E\mu_{pq}=\frac{\sum_{\sigma}\exp\left(\beta\sum_{x\sim y}(-1)^{\varphi_{pq}(xy)}\sigma_{x}\sigma_{y}\right)}{\sum_{\sigma}\exp\left(\beta\sum_{x\sim y}\sigma_{x}\sigma_{y}\right)}\\
=\frac{\sum_{\sigma}\prod_{x\sim y}\left(1+(-1)^{\varphi_{pq}(xy)}\sigma_{x}\sigma_{y}\alpha\right)}{\sum_{\sigma}\prod_{x\sim y}\left(1+\sigma_{x}\sigma_{y}\alpha\right)}=\frac{\sum_{\xi\in\EveSub{\torus}}(-1)^{\varphi_{pq}(\xi)}\alpha^{|\xi|}}{\sum_{\xi\in\EveSub{\torus}}\alpha^{|\xi|}}.
\end{multline*}
To compute $\E\mu^{(ij)},$ we recall that $\varphi_{pq}(\xi)=p\varphi_{10}(\xi)+q\varphi_{01}(\xi)\,\mod2$
and note that 
\begin{multline*}
\sum_{p,q}(-1)^{(1-i)p+(1-j)q+pq+p\varphi_{10}(\xi)+q\varphi_{01}(\xi)}\\
=(-1)^{q_{ij}(\xi)+(1-i)(1-j)}\sum_{p,q}(-1)^{(p+\varphi_{01}(\xi)+1-j)(q+\varphi_{10}(\xi)+1-i)}\\
=(-1)^{q_{ij}(\xi)+(1-i)(1-j)}\sum_{p,q}(-1)^{pq}=2(-1)^{q_{ij}(\xi)}(-1)^{(1-i)(1-j)}.
\end{multline*}
Therefore, plugging the above formula for $\E\mu_{pq}$ into (\ref{eq: def_mu^ij})
and taking into account (\ref{eq: lem_qf}) yields
\[
\E\mu^{(ij)}=\frac{1}{4}(-1)^{(1-i)(1-j)}\frac{2\sum_{\xi\in\EveSub{\torus}}(-1)^{q_{ij}(\xi)}\alpha^{|\xi|}}{\sum_{\xi\in\EveSub{\torus}}\alpha^{|\xi|}}=\frac{1}{4}\frac{2Z^{(ij)}}{\frac{1}{2}\left(Z^{(01)}+Z^{(10)}+Z^{(11)}\right)},
\]
where we have used that $Z^{(00)}=0$ and $(-1)^{(1-i)(1-j)}=1$ unless
$i=j=0$. For the asymptotics, we use $Z^{(ij)}=\sqrt{\KW^{ij}}=2^{-|\T^{\delta}|/2}\alpha_{c}^{-|\torus|}\sqrt{\det\Delta_{\delta}^{ij}}$,
and then use asymptotics of these determinants computed in the course
of the proof of Corollary~\ref{THEOREM1}.
\end{proof}

\section{Scaling limits of the fermionic observables}

\label{sec: Scaling-limits }In what follow, we define the continuous
limits of the observable $F_{e_{1}\dots e_{k}}^{(ij)}(a,z).$ These
limits will depend on $\eta_{a},$ i. e., the orientation of the corner
$a$ and the choice of the sign of the square root in (\ref{eq: Dirac}).
Since there are only $8$ options for $\eta_{a},$ we will from now
on assume it fixed. 

It will be convenient to use physics notation for the Pfaffian: for
symbols $\Op_{1},$$\dots,$$\Op_{2N},$ and a label $\diamond$ to
distinguish between different anti-periodicity ``sectors'', if a
$2N\times2N$ antisymmetric matrix $M$ is given whose entries are
denoted by $M_{n,m}=\ccor{\Op_{n}\Op_{m}}^{\diamond},$ we denote
\[
\ccor{\Op_{1}\dots\Op_{2n}}^{\diamond}:=\Pf M=\Pf\ccor{\Op_{n}\Op_{m}}_{1\leq n,m\leq2N}^{\diamond}.
\]

Recall that, given $\omega_{1,2}$, the Weierstrass $\zeta$-function
$\zeta_{\omega_{1},\omega_{2}}$ is the unique odd function that has
a simple pole of residue $1$ at the origin, and such that its derivative
is doubly periodic (in fact, $-\zeta'(z)=\wp(z),$ where $\wp(z)$
is the Weierstrass $\wp$ function). Thus, $\zeta(z)$ is not an elliptic
function, but has periodicity property $\zeta(z+\omega_{12})=\zeta(z)+c_{12}$
for some constants $c_{12}.$ However, a linear combination $\sum\alpha_{i}\zeta(z-\beta_{i})$
is an elliptic function provided that $\sum\alpha_{i}=0.$ We
denote by $\cs_{\omega_{1},\omega_{2}}(z),$ (respectively, $\ns_{\omega_{1},\omega_{2}}(z),$
$\ds_{\omega_{1},\omega_{2}}(z)$) the unique meromorphic function
on $\hat{\T}$ with four simple poles, including one with residue $1$ at the origin, and
satisfying the anti-periodicity relations (\ref{eq: f_ij_periodicity})
below with $(ij)=(01)$ (respectively, $(ij)=(10)$, $(ij)=(11)$). We have $\cs_{\omega_{1},\omega_{2}}(z)=\frac{2K}{\omega_{1}}\cdot\cs\left(\frac{2K}{\omega_{1}}\left(z-a\right),k\right)$
in the notation of \cite[Section 22]{NIST:DLMF}, and similarly for
$\ns,\ds$, where the elliptic modulus $\EllMod$ and the complete
elliptic integral $K$ are given by 
\begin{align*}
\EllMod:=\left(\frac{\JacThe 2{\ModPar}}{\JacThe 3{\ModPar}}\right)^{2},\quad K=\frac{\pi}{2}\theta_{3}^{2}(\tau).
\end{align*}
 
\begin{defn}
\label{def: fmp}Given a continuous torus $\T$ and distinct points
$a,e_{1},\dots,e_{k}\in\T$, for $(ij)\neq(00),$ we define 
\begin{equation}
f_{e_{1}\dots e_{k}}^{(ij)}(a,z)=\begin{cases}
\frac{\PartFun{ij}}{\PFTotal}\overline{\eta}_{a}\i^{k}\ccor{\psi_{e_{1}}\psi_{e_{1}}^{\star}\dots\psi_{e_{k}}\psi_{e_{k}}^{\star}\psi_{z}\psi_{a}}^{(ij)} & k\text{ even},\\
\frac{\PartFun{ij}}{\PFTotal}\eta_{a}\i^{k}\ccor{\psi_{e_{1}}\psi_{e_{1}}^{\star}\dots\psi_{e_{k}}\psi_{e_{k}}^{\star}\psi_{z}\psi_{a}^{\star}}^{(ij)} & k\text{ odd},
\end{cases}\label{eq: f_correlation}
\end{equation}
where $\ccor{\psi_{w}\psi_{\hat{w}}^{\star}}^{(ij)}\equiv0,$ $\ccor{\psi_{w}^{\star}\psi_{\hat{w}}^{\star}}^{(ij)}=\overline{\ccor{\psi_{w}\psi_{\hat{w}}}^{(ij)}}$,
and 
\begin{eqnarray}
\ccor{\psi_{w}\psi_{\hat{w}}}^{(01)} & = & \cs_{\omega_{1},\omega_{2}}(w-\hat{w}),\label{eq: ferm_ellipt}\\
\ccor{\psi_{w}\psi_{\hat{w}}}^{(10)} & = & \ns_{\omega_{1},\omega_{2}}(w-\hat{w}),\label{eq: ferm_ellipt_2}\\
\ccor{\psi_{w}\psi_{\hat{w}}}^{(11)} & = & \ds_{\omega_{1},\omega_{2}}(w-\hat{w}).\label{eq: ferm_ellipt_3}
\end{eqnarray}
\end{defn}

\begin{defn}
\label{def: fpm_00}Given data as above, we also define 
\begin{equation}
f_{e_{1}\dots e_{k}}^{(00)}(a,z)=\begin{cases}
\bar{\eta}_{a}\i^{k}\ccor{\psi_{e_{1}}\psi_{e_{1}}^{\star}\dots\psi_{e_{k}}\psi_{e_{k}}^{\star}\psi_{z}\psi_{a}}^{(00)}, & k\text{ odd,}\\
\eta_{a}\i^{k}\ccor{\psi_{e_{1}}\psi_{e_{1}}^{\star}\dots\psi_{e_{k}}\psi_{e_{k}}^{\star}\psi_{z}\psi_{a}^{\star}}^{(00)}, & k\ \text{even,}
\end{cases}\label{eq: f_corr_oo}
\end{equation}
where $\ccor{\psi_{e_{n}}^{\star}\psi_{e_{m}}^{\star}}^{(00)}=\overline{\zeta(e_{n}-e_{m})},$
\begin{equation}
\ccor{\psi_{e_{n}}\psi_{e_{m}}^{\star}}^{(00)}\equiv-\pi\i\frac{(\im\tau)^{\frac{1}{2}}|\theta_{2}\theta_{3}\theta_{4}|}{|\theta_{2}|+|\theta_{3}|+|\theta_{4}|}\cdot\frac{1}{|\T|^{\frac{1}{2}}},\label{eq: En^00_e}
\end{equation}
and $\ccor{\psi_{e_{n}}\psi_{e_{m}}}^{(00)}=\zeta(e_{n}-e_{m});$
here $\zeta$ denotes the Weierstrass $\zeta$-function.
\end{defn}

\begin{defn}
\label{def: en}Given data as above, we define for $(ij)\neq(00),$
\begin{gather*}
\En{e_{1}\dots e_{k}}{ij}:=\begin{cases}
\frac{\PartFun{ij}}{\PFTotal}\i^{k}\ccor{\psi_{e_{1}}\psi_{e_{1}}^{\star}\dots\psi_{e_{k}}\psi_{e_{k}}^{\star}}^{(ij)} & k\text{ even,}\\
0 & k\text{ odd},
\end{cases}
\end{gather*}
and 
\[
\En{e_{1}\dots e_{k}}{00}=\begin{cases}
\i^{k}\ccor{\psi_{e_{1}}\psi_{e_{1}}^{\star}\dots\psi_{e_{k}}\psi_{e_{k}}^{\star}}^{(00)}, & k\text{ odd},\\
0 & k\text{ even.}
\end{cases}
\]
\end{defn}
The above definitions extend to $k=0$, setting the empty Pfaffian to $1$.
\begin{prop}
\label{prop: expansions}The quantity $\fmp{ij}{e_{1}\dots e_{k}}{a,\cdot}$
is a meromorphic function on $\hat{\T}$ satisfying, for $p,q\in\{0,1\},$
\begin{equation}
\fmp{ij}{e_{1}\dots e_{k}}{a,z+p\omega_{1}+q\omega_{2}}=(-1)^{ip+jq}\fmp{ij}{e_{1}\dots e_{k}}{a,z}\label{eq: f_ij_periodicity}
\end{equation}
Its poles are simple and located at $e_{1},\dots,e_{k},a$, and 

\begin{eqnarray}
\fmp{ij}{e_{1}\dots e_{k}}{a,\cdot} & = & \i\frac{\overline{f_{e_{1}\dots\hat{e}_{m}\dots e_{k}}^{(ij)}(a,e_{m})}}{z-e_{m}}+\oo 1,\quad z\to e_{m},\label{eq: f_exp_em}\\
\fmp{ij}{e_{1}\dots e_{k}}{a,z} & = & \frac{\overline{\eta}_{a}\En{e_{1}\dots e_{k}}{ij}}{z-a}-\i\eta_{a}\cdot\En{e_{1}\dots e_{k}a}{ij}+\oo 1,\quad z\to a.\label{eq: f_exp_a}
\end{eqnarray}
\end{prop}

\begin{rem}
\label{rem: uniqueness}The equations (\ref{eq: f_ij_periodicity}–\ref{eq: f_exp_a})
give an \emph{overdetermined} set of conditions that identify $\fmp{ij}{e_{1}\dots e_{k}}{a,z}$
and $\En{e_{1}\dots e_{k}}{ij}$ uniquely by induction, given $\En{}{ij}$
for $(ij)\neq(00)$ and a constant $\En e{00}.$ Indeed, to see that
$f_{e_{1}\dots e_{k}}^{(ij)}(a,\cdot)$ is uniquely determined, note
that if two functions $f,g$ both satisfy (\ref{eq: f_ij_periodicity}–\ref{eq: f_exp_a}),
then, by induction hypothesis, their difference is holomorphic everywhere
on $\hat{\T}$ and vanishes, say, at $e_{1}$, hence it is zero. In
its turn, $f_{e_{1}\dots e_{k}}^{(ij)}$ uniquely determines $\En{e_{1}\dots e_{k+1}}{ij}$
by (\ref{eq: f_exp_a}). For $(ij)\neq(00),$ the vanishing of constant
terms in (\ref{eq: f_exp_em}) is not needed to prove uniqueness;
for $(ij)=(00),$ we need it just for one $m$. We note that $f^{(00)}(a,z)=-\i\eta_{a}\En e{00}$
(it is an elliptic function with at most one simple pole, hence a
constant), and, for $(i,j)\neq(0,0)$, 
\[
f^{(ij)}=\bar{\eta}_{a}\En{}{ij}\ccor{\psi_{z}\psi_{a}}^{(ij)},
\]
 where the latter is given by (\ref{eq: ferm_ellipt}–\ref{eq: ferm_ellipt_3})
\end{rem}

\begin{proof}
In all cases, $\ccor{\psi_{z}\psi_{w}}^{\diamond}=(z-w)^{-1}+\oo 1,$
since it is an odd function of $z-w$ with a simple pole of residue
$1$ at the origin; also, complex conjugating the Pfaffian amounts
to replacing $\psi\longleftrightarrow\psi^{\star}.$ Expanding the
Pfaffian and using the convention $\ccor{\psi_{e_{1}}\psi_{e_{1}}}^{\diamond}=0,$ we obtain as $z\to e_{1}$ 
\begin{multline*}
\ccor{\psi_{e_{1}}\psi_{e_{1}}^{\star}\dots\psi_{e_{k}}\psi_{e_{k}}^{\star}\psi_{z}\psi_{a}}^{\diamond}\\
=\ccor{\psi_{z}\psi_{e_{1}}}^{\diamond}\ccor{\psi_{e_{1}}^{\star}\dots\psi_{e_{k}}\psi_{e_{k}}^{\star}\psi_{a}}^{\diamond}+\ccor{\psi_{e_{1}}\psi_{e_{1}}^{\star}\dots\psi_{e_{k-1}}\psi_{e_{k-1}}^{\star}\psi_{e_{1}}\psi_{a}}^{\diamond}+\oo 1\\
=(z-e_{1})^{-1}\ccor{\psi_{e_{1}}^{\star}\dots\psi_{e_{k}}\psi_{e_{k}}^{\star}\psi_{a}}^{\diamond}+\oo 1\\
=(z-e_{1})^{-1}(-1)^{k-1}\overline{\ccor{\psi_{e_{1}}\dots\psi_{e_{k}}\psi_{e_{k}}^{\star}\psi_{a}^{\star}}^{\diamond}}+\oo 1,
\end{multline*}
and similarly for $m=2,\dots,k$ and with $\psi_{a}$ replaced with
$\psi_{a}^{\star}.$ This proves (\ref{eq: f_exp_em}). Similarly,
as $z\to a$, we have 
\begin{eqnarray*}
\ccor{\psi_{e_{1}}\psi_{e_{1}}^{\star}\dots\psi_{e_{k}}\psi_{e_{k}}^{\star}\psi_{z}\psi_{a}}^{\diamond} & = & \ccor{\psi_{z}\psi_{a}}\ccor{\psi_{e_{1}}^{\star}\dots\psi_{e_{k}}\psi_{e_{k}}^{\star}}^{\diamond}+\oo 1,\\
\ccor{\psi_{e_{1}}\psi_{e_{1}}^{\star}\dots\psi_{e_{k}}\psi_{e_{k}}^{\star}\psi_{z}\psi_{a}^{\star}}^{\diamond} & = & \ccor{\psi_{e_{1}}\psi_{e_{1}}^{\star}\dots\psi_{e_{k}}\psi_{e_{k}}^{\star}\psi_{a}\psi_{a}^{\star}}^{\diamond}+\oo 1.
\end{eqnarray*}
The first identity proves (\ref{eq: f_exp_a}) for $(ij)\neq(00),$
$k$ even and for $(ij)=0,$ $k$ odd, while the second identity proves
it for $(ij)\neq00,$ $k$ odd and for $(ij)=(00),$ $k$ even. 

To prove (\ref{eq: f_ij_periodicity}) for $(ij)\neq(00),$ simply
expand the Pfaffian in $\psi_{z}$ and note that $\ccor{\psi_{z}\psi_{w}}^{(ij)}$
satisfies (\ref{eq: f_ij_periodicity}) for each $w=e_{1},\dots,e_{k},a.$
For $(ij)=(00),$ the Weierstrass $\zeta$-function does not satisfy
(\ref{eq: f_ij_periodicity}); as discussed above, we need to check
that the sum of the residues of $f_{e_{1}\dots e_{k}}^{(00)}(a,\cdot)$
is zero. We claim that this result follows by induction from the existence
of \emph{any} solution to (\ref{eq: f_ij_periodicity}–\ref{eq: f_exp_a})
with $-\i\En e{00}$ given by (\ref{eq: En^00_e}). Indeed, suppose
$\tilde{f}_{e_{1}\dots e_{k}}^{(00)}(a,z)$ is such a solution, and
suppose by induction hypothesis that $f_{e_{1\dots}e_{k-1}}^{(00)}(a,z)=\tilde{f}_{e_{1}\dots e_{k-1}}^{(00)}(a,z)$
for any $e_{1},\dots,e_{k-1},a$. Then, by (\ref{eq: f_exp_em}–\ref{eq: f_exp_a}),
the residues of $f_{e_{1}\dots e_{k}}^{(00)}(a,z)$ and $\tilde{f}_{e_{1}\dots e_{k}}^{(00)}$
are the same; since the latter function satisfies (\ref{eq: f_ij_periodicity}),
their sum vanishes; hence also $f_{e_{1}\dots e_{k}}^{(00)}(a,z)$
satisfies (\ref{eq: f_ij_periodicity}), and we have $f_{e_{1}\dots e_{k}}^{(00)}\equiv\tilde{f}_{e_{1}\dots e_{k}}^{(00)}$,
completing the induction step. The solution to (\ref{eq: f_ij_periodicity}–\ref{eq: f_exp_a})
exists since it is constructed as the limit of discrete observables
in the proof of Theorem \ref{thm: obs_conv} below. 
\end{proof}
At the heart of our analysis is the following convergence result:
\begin{thm}
\label{thm: obs_conv}One has, as $\delta\to0,$
\begin{equation}
\delta^{-(k+1)}F_{e_{1}\dots e_{k}}^{(ij)}(a,z)=\Conk k\cdot\proj{\eta_{z}}{\fmp{ij}{e_{1}\dots e_{k}}{a,z}}+o(1),\label{eq: conv_observable}
\end{equation}
uniformly in $z,e_{1},\dots,e_{k},a$ away from each other, where
$\Conk k=\frac{\sqrt{2}}{\pi^{k+1}}$. Moreover, 
\begin{equation}
\delta^{-k}\E\left[\mu^{(ij)}\prod_{m=1}^{k}\en_{e_{m}}\right]\to\pi^{-k}\En{e_{1}\dots e_{k}}{ij}.\label{eq: conv_energy}
\end{equation}
\end{thm}

\begin{proof}
We prove this result by induction: (\ref{eq: conv_energy})$_{k}\Longrightarrow$(\ref{eq: conv_observable})$_{k}\Longrightarrow$(\ref{eq: conv_energy})$_{k+1}$,
first separately for each $(ij)\neq(00)$ and then for $(ij)=(00).$
For the base of induction, for $(ij)\neq0$, we have (\ref{eq: En_1}–\ref{eq: En_3}),
and for $(ij)=(00),$ the induction starts with $k=1.$ Namely, we
have 
\[
\delta^{-1}\E\left[\mu^{(00)}\en_{e}\right]=\delta^{-1}\E\en_{e}-\sum_{\substack{(ij)\neq(00)}
}\delta^{-1}\E\left[\mu^{(ij)}\en_{e}\right]\to\pi^{-1}\En e{00}-\pi^{-1}\sum_{\substack{(ij)\neq(00)}
}\En e{ij},
\]
where the convergence of the first term is by Corollary (\ref{Cor: hor_only})
and the definition of $\En e{00}$, and the convergence of other three
terms will by that point have been proven. It remains to notice that
$\En e{ij}\equiv0$ for $(ij)\neq0$. Note that the proof of Corollary
\ref{Cor: hor_only} relies on Theorem \ref{thm: difference}
whose proof below does use (\ref{eq: conv_observable}), but only
in the case $(ij)\neq(00).$ Thus, the base of induction is established. 

To prove (\ref{eq: conv_observable}), we follow a general scheme
used in \cite{CHI_Mixed}, where we first identify the scaling limit
assuming precompactness, and then justify precompactness. Let us first
assume that for all $r>0$, there is a constant $C_{r}$ such that
the functions $\left|\delta^{-(k+1)}F_{e_{1}\dots e_{k}}^{(ij)}(a,\cdot)\right|$
are bounded by $C_{r}$ on the set $\T_{r}:=\T\setminus B_{r}(a)\cup B_{r}(e_{1})\cup\dots\cup B_{r}(e_{k})$,
uniformly in $\mesh$. We claim that in this case, 
\begin{align}
\delta^{-(k+1)}F_{e_{1}\dots e_{k}}^{(ij)}(a,\cdot)=\Conk k\cdot\proj{\eta_{z}}{\fmp{ij}{e_{1}\dots e_{k}}{a,z}}+\oo 1\label{EQUATIONAUXILIARYCLAIM-1}
\end{align}
uniformly on compact subsets of $\hat{\T}\setminus\{a_{pq}\}$. Indeed,
as noted above, the s-holomorphicity of $F_{e_{1}\dots e_{k}}^{(ij)}(a,z)$,
together with uniform boundedness, implies that (\ref{EQUATIONAUXILIARYCLAIM-1})
holds along a subsequence, and with \emph{some} holomorphic function
$f:\hat{\T}\setminus\{a_{pq}\}\to\C$ instead of $\fmp{ij}{e_{1}\dots e_{k}}{a,z}$.
Therefore, it is enough to show that any sub-sequential limit $f$
must be equal to $\fmp{ij}{e_{1}\dots e_{k}}{a,z}$. Clearly, $f$
must satisfy the (anti)-periodicity condition $f(z+p\omega_{1}+q\omega_{2})=(-1)^{ip+jq}f(z)$,
because of (\ref{eq: anti-symmetry}). 

Denote $U=\fmp{ij}{e_{1}\dots\hat{e}_{m}\dots e_{k}}{a,e_{m}}.$ We
have by induction hypothesis 
\begin{eqnarray*}
\delta^{-k}F_{e_{1}\dots\hat{e}_{m}\dots e_{k}}^{(ij)}(a,z_{2m}) & = & \Conk{k-1}\proj{\eta_{z_{2m}}}U+\oo 1,\\
\delta^{-k}F_{e_{1}\dots\hat{e}_{m}\dots e_{k}}^{(ij)}(a,z_{2m-1}) & = & \Conk{k-1}\proj{\eta_{z_{2m-1}}}U+\oo 1,
\end{eqnarray*}
so that, taking into account (\ref{eq: asymp_zmone}) and the identity
$\proj{\eta}U=\frac{1}{2}(U+\eta^{2}\overline{U})$, in a fixed annulus
around $B_{R}(e_{m})\setminus B_{r}(e_{m})$, we have, for the ``corrective''
term in the definition of $\tilde{F}_{k}(\cdot)$: 
\begin{multline*}
\i\bar{\eta}_{z_{2m-1}}\delta^{-k-1}\frac{\delta}{\sqrt{2}}F_{e_{1}\dots\hat{e}_{m}\dots e_{k}}^{(ij)}(a,z_{2m})\dzmone_{z_{2m-1}}^{\delta}(z)+\i\bar{\eta}_{z_{2m}}\delta^{-k-1}\frac{\delta}{\sqrt{2}}F_{e_{1}\dots\hat{e}_{m}\dots e_{k}}^{(ij)}(a,z_{2m-1})\dzmone_{z_{2m}}^{\delta}(z)\\
=\Conk{k-1}\cdot\frac{1}{\pi}\i\bar{\eta}_{z_{2m-1}}\proj{\eta_{z_{2m}}}U\proj{\eta_{z}}{\frac{\bar{\eta}_{z_{2m-1}}}{z-e_{m}}}\\
+\Conk{k-1}\cdot\frac{1}{\pi}\i\bar{\eta}_{z_{2m}}\proj{\eta_{z_{2m-1}}}U\proj{\eta_{z}}{\frac{\bar{\eta}_{z_{2m}}}{z-e_{m}}}+\oo 1\\
=\Conk{k-1}\cdot\frac{1}{\pi}\cdot\proj{\eta_{z}}{\frac{\i\bar{U}}{z-e_{m}}}+\oo 1.
\end{multline*}
Recall $\delta^{-(k+1)}\tilde{F}_{m}(\cdot)$ (\ref{eq: res_e_m})
extends to an s-holomorphic function near $e_{m}.$ We can express
its values by discrete Cauchy integral formula with contour in $B_{R}(e_{m})\setminus B_{r}(e_{m})$
and pass to the limit in that formula; this shows that $f(z)-C_{k}\cdot\frac{\i\bar{U}}{z-e_{m}}$
extends to a holomorphic function at $B_{2r}(e_{m});$ moreover, (\ref{eq: values_at_e_i})
shows that this function vanishes at $e_{m}$. Hence, $f$ must satisfy
(\ref{eq: f_exp_em}). 

The analysis near $a$ is similar. Since $\delta^{-k}\cdot\E\left[\mu^{(ij)}\prod_{i=1}^{k}\en_{e_{i}}\right]=\En{e_{1}\dots e_{k}}{ij}+\oo 1$
by induction hypothesis, we have the expansion 
\[
\delta^{-k}\E\left[\mu^{(ij)}\prod_{i=1}^{k}\en_{e_{i}}\right]\dzmone_{a}^{\delta}(\cdot)=\frac{\sqrt{2}}{\pi^{k+1}}\En{e_{1}\dots e_{k}}{ij}\cdot\proj{\eta_{z}}{\frac{\bar{\eta}_{a}}{z-a}}+\oo 1.
\]
uniformly in $B_{R}(e_{m})\setminus B_{r}(e_{m}).$ Therefore, by
the same argument as above, $f(z)-\overline{\eta}_{a}\frac{\En{e_{1}\dots e_{k}}{ij}}{z-a}$
can be analytically continued in $B_{R}(a)$, and thus satisfies (\ref{eq: f_exp_a}).
In other words, $f(\cdot)$ satisfies the defining conditions of $\fmp{ij}{e_{1}\dots e_{k}}{a,\cdot}$,
and therefore, due to Remark \ref{rem: uniqueness}, we have $f(\cdot)\equiv\fmp{ij}{e_{1}\dots e_{k}}{a,\cdot}.$

We now turn to justifying the uniform (in $\mesh$) boundedness away
form $a,$$e_{1},\dots,e_{k}$. Let us assume towards a contradiction
that there exist a (small fixed) $R>0$ such that $C_{R}^{\mesh}:=\max_{z\in\hat{\T}_{R}^{\delta}}\abs{\delta^{-1-k}F_{e_{1}\dots e_{k}}^{(ij)}(a,\cdot)}$
tends to infinity (at least along some sequence of $\mesh$). We claim
that in that case, the functions $\abs{(C_{R}^{\delta})^{-1}\delta^{-1-k}F_{e_{1}\dots e_{k}}^{(ij)}(a,\cdot)}$
are uniformly bounded on \emph{any} $\hat{\T}_{r}^{\delta}$ with
$r<R$. Indeed, 
\[
(C_{r}^{\mesh})^{-1}\left(\delta^{-1-k}F_{e_{1}\dots e_{k}}^{(ij)}(a,\cdot)-\delta^{-k-1}\E\left[\mu^{(ij)}\prod_{i=1}^{k}\en_{e_{i}}\right]\dzmone_{a}^{\delta}(\cdot)\right)
\]
is uniformly bounded on $B(a,2R)\setminus B(a,R)$ and discrete holomorphic
in $B(a,2R)$, and hence, by maximum principle, it is uniformly bounded
in the whole of $B(a,2R)$. As $\dzmone_{a}^{\delta}(z)$ is also
uniformly bounded on compact subsets of $B(a,2R)\setminus\{a\}$ and
$(C_{r}^{\mesh})^{-1}\to0$, we get the claim. Similarly, we justify
the uniform boundedness on each $B_{R}(e_{m})\setminus B_{r}(e_{m})$
for $m=1,\dots,k.$ 

Therefore, the above convergence argument can be applied verbatim
to 
\[
(C_{R}^{\delta})^{-1}\delta^{-1-k}F_{e_{1}\dots e_{k}}^{(ij)}(a,\cdot)
\]
with the conclusion that it converges uniformly on compact subsets
of $\hat{\mathbb{T}}$ to a function that is analytic in the whole
$\hat{\T}$ and vanishing, say, at $e_{1}$, that is, to $0$. This
contradicts the choice of $C_{r}^{\mesh}$. 

To derive (\ref{eq: conv_energy})$_{k+1}$, note that we have shown
above that $\tilde{F}_{a}(\cdot)$ extends to an s-holomorphic function
in a neighborhood of $a$, and that 
\[
\delta^{-(k+1)}\tilde{F}_{a}(z)=\Conk k\cdot\proj{\eta_{z}}{\fmp{ij}{e_{1}\dots e_{k}}{a,z}-\En{e_{1}\dots e_{k}}{ij}\frac{\bar{\eta}_{a}}{z-a}}+\oo 1,
\]
uniformly in a neighborhood of $a,$ where the function inside the
projection is analytically continued. It remains to take into account
(\ref{eq: F_value_near}–\ref{eq: F_value_near_2}) and (\ref{eq: P_a_close})
to note that 
\[
\tilde{F}_{a}(a_{L,R})=-\i\eta_{a}\cdot\sqrt{2}\cdot\E\left[\mu^{(ij)}\en_{e_{L,R}}\prod_{m=1}^{k}\en_{m}\right].
\]
\end{proof}
\begin{rem}
The planar domain counterpart of this theorem is a particular case
of (\cite[Theorem 2]{CHI_Mixed}), asserting that
\begin{equation}
\delta^{-k-1}\E\left[\en_{e_{1}}\dots\en_{e_{k}}\psi_{a}^{\eta_{a}}\psi_{z}\right]=\Ceps^{k}C_{\psi}^{2}\eta_{a}\proj{\eta_{z}}{\ccor{\en_{e_{1}}\dots\en_{e_{k}}\psi_{a}^{\eta_{a}}\psi_{z}}}+\oo 1.\label{eq: CHI}
\end{equation}
The expansions of the right-hand side at $z=e_{m}$ and $z=a$ can
be read off the fusion rules \cite[(6.2–6.4)]{CHI_Mixed}. In the
case $(ij)\neq(00),$ the same convergence proof as in \cite{CHI_Mixed},
by expressing $F_{e_{1}\dots e_{k}}^{(ij)}(a,z)$ as a Pfaffian of
two-point correlations, and then passing to the limit term-by-term,
could have been carried on. We were unable to do it for the $(00)$
sector, due to the lack of combinatorial analog of $\ccor{\psi_{z}\psi_{a}}^{(00)}.$
By the formalism of \cite[Section 5.2]{CHI_Mixed}, in the RHS of
(\ref{eq: CHI}) we can expand $\psi_{a}^{\eta_{a}}=\bar{\eta}_{a}\psi_{a}+\eta_{a}\psi_{a}^{\star}$;
in the torus case, only one of these terms will contribute for each
sector depending on pairity of $k$, which is reflected in (\ref{eq: f_correlation},
\ref{eq: f_corr_oo}).
\end{rem}

\section{Proofs of Theorem \ref{thm: difference} and Theorem \ref{thm: multipoint}.}

In order to prove Theorem \ref{thm: difference}, we need the one
more Lemma. Denote 
\begin{align*}
g^{(ij)}(a,z):=\fmp{ij}{}{a,z}-\frac{\En{}{ij}\bar{\eta}_{a}}{z-a};\\
G^{(ij)}(a,z):=\delta^{-1}F^{(ij)}(a,z)-\E\left[\mu^{(ij)}\right]\dzmone_{a}^{\delta}(z).
\end{align*}
It is a standard fact that if discrete holomorphic functions converge
uniformly in a ball $B(a,r)$, then so do their discretized derivatives.
This is what it means for our case: 
\begin{lem}
\label{lem: diff_obs_asymp}One has, as $\delta\to0$, for $(ij)\neq(00),$
\begin{align*}
\delta^{-2}\left(F^{(ij)}(a,\anw)-F^{(ij)}(a,\ase)\right)\to\frac{2}{\pi}\i\eta_{a}\re\left[\i\bar{\eta}_{a}^{3}\left.\pa_{z}g^{(ij)}(a,z)\right|_{z=a}\right].
\end{align*}
\end{lem}

\begin{proof}
First of all, note that by (\ref{eq: P_a_close}), we have 
\[
\delta^{-2}\left(F^{(ij)}(a,\anw)-F^{(ij)}(a,\ase)\right)=\delta^{-1}\left(G^{(ij)}(a,a_{L})-G^{(ij)}(a,a_{R})\right).
\]
Now, by (\ref{eq: res_a}), $G^{(ij)}(a,\cdot)$ is discrete holomorphic
in a (fixed) neighborhood of $a$, its restriction to the sub-lattice
$\{u\in\Cgr:\eta_{u}=\eta_{a_{L}}=\eta_{a_{R}}=\i\eta_{a}\}$ is a
discrete harmonic function, as it was shown in the proof of Theorem
\ref{thm: obs_conv} that converges to $\frac{\sqrt{2}}{\pi}\proj{\i\eta_{a}}{g^{(ij)}(z)}$.
It is well known that this implies convergence of its (normalized)
finite difference to the corresponding derivative. Denote $a_{L}-a_{R}:=a_{LR}=-\sqrt{2}\delta\bar{\eta}_{a}^{2}$,
we thus have 
\begin{multline*}
G^{(ij)}(a,\anw)-G^{(ij)}(a,\ase)\\
=\frac{\sqrt{2}}{\pi}\left(a_{LR}\left.\pa_{z}\proj{\i\eta_{a}}{g^{(ij)}(a,z)}\right|_{z=a}+\overline{a_{LR}}\left.\pa_{\overline{z}}\proj{\i\eta_{a}}{g^{(ij)}(a,z)}\right|_{z=a}\right)+\oo{\delta}.
\end{multline*}
Since $\proj{\eta}x=\frac{1}{2}\left(x+\eta^{2}\bar{x}\right),$ $a_{LR}=-\sqrt{2}\delta\bar{\eta}_{a}^{2}$,
and $g^{(ij)}(a,\cdot)$ is holomorphic, using the notation $X=\left.\pa_{z}g^{(ij)}(a,z)\right|_{z=a},$
this can be further simplified as 
\begin{multline*}
a_{LR}\left.\pa_{z}\proj{\i\eta_{a}}{g^{(ij)}(a,z)}\right|_{z=a}+\overline{a_{LR}}\left.\pa_{\overline{z}}\proj{\i\eta_{a}}{g^{(ij)}(a,z)}\right|_{z=a}\\
=-\frac{\sqrt{2}}{2}\delta\left(\bar{\eta}_{a}^{2}X-\eta_{a}^{4}\bar{X}\right)=\sqrt{2}\delta\i\eta_{a}\frac{1}{2\i}\left(\eta_{a}^{3}\bar{X}-\bar{\eta}_{a}^{3}X\right)\\
=\sqrt{2}\delta\i\eta_{a}\im[\eta_{a}^{3}\bar{X}]=\sqrt{2}\delta\i\eta_{a}\re[\i\bar{\eta}_{a}^{3}X].
\end{multline*}
\end{proof}
We are in the position to prove Theorem \ref{thm: difference}:
\begin{proof}[Proof of Theorem \ref{thm: difference}.]
 Note that the correlations $\E[\mu^{(ij)}(\en_{V}-\en_{H})]$ can
be expressed, using (\ref{eq: F_ij_at_LR}), as 
\begin{equation}
\E\left[\mu^{(ij)}(\en_{e_{L}}-\en_{e_{R}})\right]=-\frac{1}{\i\sqrt{2}\eta_{a}}\left(F^{(ij)}(a,a_{L})-F^{(ij)}(a,a_{R})\right).\label{eq: en_mu}
\end{equation}
In fact, this is zero for $(ij)=(00)$ by (\ref{eq: F^00 constant}).
Choose $a$ so that $\frac{a^{\bullet}-a^{\circ}}{|a^{\bullet}-a^{\circ}|}=e^{\i\frac{\pi}{4}}$,
i. e., $\eta_{a}=e^{\i\frac{\pi}{8}}$; then $\en_{e_{L}}=\en_{V},$
$\en_{e_{R}}=\en_{H}.$ Since $\sum_{i,j\in\{0,1\}}(-1)^{(1-i)p+(1-j)q+pq}=4\ind_{p=q=0}$
we have 
\begin{equation}
\sum_{(ij)}\mu^{(ij)}=\mu_{00}=1.\label{eq: sum_mu}
\end{equation}
Therefore, summing (\ref{eq: en_mu}) and applying Lemma \ref{lem: diff_obs_asymp}
yields
\begin{align}
\DifEneDen & =\delta^{2}\sum_{(ij)\neq(00)}\frac{\delta^{-2}}{\i\sqrt{2}\eta_{a}}\left(F^{(ij)}(a,\anw)-F^{(ij)}(a,\ase)\right).\label{eq: diff_intermediate}\\
 & =\delta^{2}\cdot\frac{\sqrt{2}}{\pi}\cdot\sum_{(i,j)\neq(0,0)}\left(\re\left[\i\bar{\eta}_{a}^{3}\pa_{z}g^{(ij)}(a,z)|_{z=a}\right]+\oo 1\right).\nonumber 
\end{align}
It remains to compute explicitly the derivatives. Jacobian elliptic
functions have series expansions
\begin{align*}
\cs\left(u,\EllMod\right) & =\frac{1}{u}+\left(\frac{-1}{3}+\frac{1}{6}\EllMod^{2}\right)u+\OO{u^{3}},\\
\ns\left(u,\EllMod\right) & =\frac{1}{u}+\frac{1}{6}\left(1+\EllMod^{2}\right)u+\OO{u^{3}},\\
\ds\left(u,\EllMod\right) & =\frac{1}{u}+\left(\frac{1}{6}-\frac{1}{3}\EllMod^{2}\right)u+\OO{u^{3}}
\end{align*}
and plugging them into the formula for $f^{(ij)}(a,\cdot)$ yields
\[
\left.\partial_{z}g^{(ij)}(a,z)\right|_{z=a}=\bar{\eta}_{a}\En{}{ij}\frac{4K^{2}}{\omega_{1}^{2}}\cdot\begin{cases}
\frac{1}{6}\EllMod^{2}-\frac{1}{3}\\
\frac{1}{6}\left(1+\EllMod^{2}\right)\\
\frac{1}{6}-\frac{1}{3}\EllMod^{2}
\end{cases}=\bar{\eta}_{a}\En{}{ij}\frac{\pi^{2}}{24\omega_{1}^{2}}\cdot\begin{cases}
\theta_{2}^{4}-2\theta_{3}^{4}, & (ij)=(01),\\
\theta_{2}^{4}+\theta_{3}^{4}, & (ij)=(10),\\
\theta_{3}^{4}-2\theta_{2}^{4}, & (ij)=(11).
\end{cases}
\]
 Plugging into (\ref{eq: diff_intermediate}) concludes the proof. 
\end{proof}
\begin{proof}[Proof of Theorem \ref{thm: multipoint}.]
 By (\ref{eq: sum_mu}) and (\ref{eq: conv_energy}), we have 
\[
\delta^{-k}\E\left[\prod_{m=1}^{k}\en_{m}\right]=\sum_{(ij)}\delta^{-k}\E\left[\mu^{(ij)}\prod_{m=1}^{k}\en_{m}\right]\to\pi^{-k}\sum_{(ij)}\En{e_{1}\dots e_{k}}{ij}.
\]
By Definition \ref{def: en}, only $\En{e_{1}\dots e_{k}}{00}$ contributes
for odd $k;$ chasing the definitions leads directly to (\ref{eq: corr_en_odd}).
For even $k$, only $\En{e_{1}\dots e_{k}}{01},\En{e_{1}\dots e_{k}}{10},\En{e_{1}\dots e_{k}}{11}$
contribute, and we have 
\begin{multline*}
\En{e_{1}\dots e_{k}}{ij}=\frac{\PartFun{ij}}{\PFTotal}(-1)^{\frac{k}{2}}\ccor{\psi_{e_{1}}\psi_{e_{1}}^{\star}\dots\psi_{e_{k}}\psi_{e_{k}}^{\star}}^{(ij)}\\
=\frac{\PartFun{ij}}{\PFTotal}(-1)^{\frac{k}{2}}(-1)^{\frac{k(k-1)}{2}}\ccor{\psi_{e_{1}}\dots\psi_{e_{k}}\psi_{e_{1}}^{\star}\dots\psi_{e_{k}}^{\star}}^{(ij)}\\
=\frac{\PartFun{ij}}{\PFTotal}\ccor{\psi_{e_{1}}\dots\psi_{e_{k}}}^{(ij)}\ccor{\psi_{e_{1}}^{\star}\dots\psi_{e_{k}}^{\star}}^{(ij)}=\frac{\PartFun{ij}}{\PFTotal}\left|\ccor{\psi_{e_{1}}\dots\psi_{e_{k}}}^{(ij)}\right|^{2}.
\end{multline*}
\end{proof}

\section{Theorem \ref{thm: sum} for the triangular lattice}

\label{sec: triangular}In this section, we prove an analog of Theorem
\ref{thm: sum} for triangular lattice $\H:=\{m+e^{\frac{\i\pi}{3}}n:m,n\in\Z\}$;
a similar argument (or duality) can be applied to treat the hexagonal
lattice. We follow the argument for the case of the square lattice.
The energy observable on the triangular lattice is defined by 
\[
\en_{(xy)}=\sigma_{x}\sigma_{y}-\bar{\en},
\]
where $\bar{\en}=\frac{2}{3}.$ The critical value of $\alpha$ is $\alpha_{\tri}=2-\sqrt{3}$. There are three types of edges on
the lattice, of which we choose representatives $e_{0},$ $e_{1}$,
$e_{2},$ where $e_{k}$ is aligned with $e^{\frac{\pi\i}{3}k}.$
Our result reads as follows: 
\begin{thm}
\label{thm: thm_1_triang}For the critical Ising model on 
$\torus=\delta\H/\Lambda^{\delta}$, we have 
\[
\E\en_{e_{0}}+\E\en_{e_{1}}+\E\en_{e_{2}}=6\sqrt{2}\frac{\sqrt{\det^{\star}\TwiLap{00}_{\delta}}}{\sqrt{\det\TwiLap{10}_{\delta}}+\sqrt{\det\TwiLap{01}_{\delta}}+\sqrt{\det\TwiLap{11}_{\delta}}}\frac{1}{|\torus|}.
\]
where $\det^{\star}$ denotes the product of all non-zero eigenvalues. 
\end{thm}

\begin{rem}
Since the asymptotics of the determinant of the Laplacian has been
worked out for arbitrary doubly-periodic lattices \cite{izyurov2020asymptotics},
this does lead immediately to the analog of Corollary \ref{THEOREM1},
obtained by other methods in \cite{SALAS2}. Other part of the paper
extend to $\H$ as well: the discrete holomorphicity techniques of
Section \ref{sec: disc_hol} are known to extend to a larger class
of isoradial lattices with critical weights, see \cite{CHELKAKSMIRNOV,ChelkakIzyurovMahfouf}.
\end{rem}

As in the square lattice case, we start with the high-temperature
expansion identity
\begin{equation}
\E\en_{e}=\frac{1}{Z^{I}}\sum_{\xi\in\EveSub{\torus}}b(e,\xi)\alpha^{|\xi|},\label{eq: ht_triang}
\end{equation}
where $b(e,\xi)=(\alpha^{-1}-\bar{\en})\ind_{e\in\xi}+(\alpha-\bar{\en})\ind_{e\notin\xi}.$ 

Denote the triangular lattice by $\H$, and define 
\begin{align*}
v^{\tri}(\alpha,q) & :=(1+3\alpha^{2}+8\alpha^{3}+3\alpha^{4}+\alpha^{6})\\
 & +(3\alpha^{3}-\alpha-\alpha^{5})\left(z+w+\frac{1}{z}+\frac{1}{w}+\frac{z}{w}+\frac{w}{z}\right),
\end{align*}
where, as in section 3, $z\left(q\right)=\exp\left(2\pi\i\re q\right)$
and $w\left(q\right)=\exp\left(2\pi\i\im q\right)$. The Kac–Ward
determinants in this case read, as before
\[
\det\KW^{ij}=\prod_{q\in\sfrac{\H}{\TorLat^{*}}+\TwiKacWarForShiVec{ij}}v^{\tri}(\alpha,q).
\]
As in the case of a square lattice, the determinant of the Laplacian
is related to the critical Kac–Ward determinant:
for the Laplacian $\Delta^{ij}$ on the triangular lattice, the eigenvalues
are 
\[
c_{\tri}v^{\tri}(\alpha_{\tri},q),\quad q\in\sfrac{\H}{\TorLat^{*}}+\TwiKacWarForShiVec{ij},
\]
with $c_{\tri}=\frac{1}{12(-26+15\sqrt{3})}$. 

Denoting 
\[
B^{(ij)}(e):=\sum_{\xi\in\EveSub{\torus}}b(e,\xi)\alpha^{|\xi|}\loopfactor{ij},
\]
we have the following analog of Lemma \ref{lem: B_ij}: 
\begin{lem}
\label{lem: B_ij-1}One has, for $\alpha=\alpha_{\tri},$ 
\begin{align*}
\DisSpiObsVal{ij}{e_{0}}+\DisSpiObsVal{ij}{e_{1}}+\DisSpiObsVal{ij}{e_{2}}=\begin{cases}
12(11-5\sqrt{3})c_{\tri}^{\frac{1-|\T^{\delta}|}{2}}\sqrt{\mathrm{det}^{\star}\Delta^{00}}\frac{1}{|\torus|}, & i=j=0;\\
0, & \text{otherwise}.
\end{cases}
\end{align*}
\end{lem}

\begin{proof}
As in the square lattice case, we have 
\begin{equation}
\frac{d}{d\alpha}Z^{(ij)}=\frac{|\T^{\delta}|}{\alpha}\sum_{\xi\in\EveSub{\torus}}\left(\indicator_{e_{0}\in\xi}+\ind_{e_{1}\in\xi}+\ind_{e_{2}\in\xi}\right)\CriPar^{\abs{\xi}}\loopfactor{ij}\label{eq: dZ_hex_1}
\end{equation}
The following identities are straightforward to check: 
\begin{align}
\left.\frac{d}{d\alpha}v^{\tri}(\alpha,q)\right|_{\alpha=\alpha_{\tri}} & \equiv\left(2+\frac{1}{\sqrt{3}}\right)v^{\tri}(\alpha_{\tri},q).\label{eq: deriv_hex}\\
\left.\frac{d^{2}}{d\alpha{}^{2}}v^{\tri}(\alpha,0)\right|_{\alpha=\alpha_{\tri}} & =144(2-\sqrt{3}).\label{eq: deriv_hex_second}
\end{align}
on the other hand, for $(ij)\neq(00),$ we have $ \det\KW^{ij}\neq0$,
and therefore (\ref{eq: deriv_hex}) gives

\begin{equation}
\left.\frac{d}{d\alpha}Z^{(ij)}\right|_{\alpha=\alpha_{\tri}}=\left.\frac{d}{d\alpha}\sqrt{\det\KW^{ij}}\right|_{\alpha=\alpha_{\tri}}=\left.\frac{1}{2\sqrt{\det\KW^{ij}}}\frac{d}{d\alpha}\det\KW^{ij}\right|_{\alpha=\alpha_{\tri}}=\left.\frac{c}{2}|\T^{\delta}|Z_{ij}\right|_{\alpha=\alpha_{\tri}}\label{eq: dz_hex_2}
\end{equation}
where $c=2+\frac{1}{\sqrt{3}}.$ Subtracting (\ref{eq: dZ_hex_1})
and (\ref{eq: dz_hex_2}), we get 
\[
\sum_{\xi\in\EveSub{\torus}}\left(\hat{b}(e_{0},\xi)+\hat{b}(e_{1},\xi)+\hat{b}(e_{2},\xi)\right)\CriPar^{\abs{\xi}}\loopfactor{ij}=0,
\]
where $\hat{b}(e,\xi)=\left(\frac{c}{6}-\frac{1}{\alpha}\right)\ind_{e\in\xi}+\frac{c}{6}\ind_{e\notin\xi}.$
It is now a matter of simple algebra to check that $b(e,\xi)=(6-4\sqrt{3})\hat{b}(e,\xi),$
concluding the proof in the case $(i,j)\neq(0,0).$ 

For $(i,j)=(0,0),$ we have $Z^{(00)}=0$, therefore, subtracting
$0=\frac{c}{2}|\T^{\delta}|Z^{(00)}$ from the right-hand side of
(\ref{eq: dZ_hex_1}) yields
\[
\left.\frac{d}{d\alpha}Z_{ij}\right|_{\alpha=\alpha_{\tri}}=-\sum_{\xi\in\EveSub{\torus}}\left(\hat{b}(e_{0},\xi)+\hat{b}(e_{1},\xi)+\hat{b}(e_{2},\xi)\right)\CriPar_{\tri}^{\abs{\xi}}\loopfactor{ij}.
\]
Now, since $Z_{00}(\alpha)=\sqrt{\det\KW^{00}(\alpha)}$, and because
of (\ref{eq: deriv_hex}–\ref{eq: deriv_hex_second}), we have 
\[
\det\KW^{00}(\alpha)=\frac{1}{2}144(2-\sqrt{3})(\alpha-\alpha_{c})^{2}\prod_{q\in\sfrac{\H}{\TorLat^{*}}\setminus\{0\}}\TwiKacWarWeiDetFac{\alpha}q+\oo{(\alpha-\alpha_{c})^{2}},\quad\alpha\to\alpha_{c},
\]
therefore, we arrive at 
\begin{multline*}
\left.\frac{d}{d\alpha}Z_{ij}\right|_{\alpha=\alpha_{\tri}}=12\sqrt{1-\frac{\sqrt{3}}{2}}\left(\prod_{q\in\sfrac{\H}{\TorLat^{*}}\setminus\{0\}}\TwiKacWarWeiDetFac{\alpha}q\right)^{\frac{1}{2}}=6(\sqrt{3}-1)\left(\prod_{q\in\sfrac{\H}{\TorLat^{*}}\setminus\{0\}}\TwiKacWarWeiDetFac{\alpha}q\right)^{\frac{1}{2}}.\\
=6(\sqrt{3}-1)c_{\tri}^{\frac{1-|\T^{\delta}|}{2}}\sqrt{\mathrm{\det}^{\star}\Delta_{\delta}^{00}}.
\end{multline*}
Putting everything together, we arrive at 
\[
\sum_{\xi\in\EveSub{\torus}}\left(b(e_{0},\xi)+b(e_{1},\xi)+b(e_{2},\xi)\right)\CriPar_{\tri}^{\abs{\xi}}\loopfactor{ij}=(4\sqrt{3}-6)6(\sqrt{3}-1)c_{\tri}^{\frac{1-|\T^{\delta}|}{2}}\sqrt{\mathrm{\det}^{\star}\Delta_{\delta}^{00}},
\]
as required.
\end{proof}
\begin{proof}[Proof of Theorem \ref{thm: thm_1_triang}]
 As in the the case of the square lattice, (\ref{eq: ht_triang})
implies that 
\[
\E\en_{e}=\frac{1}{2Z^{I}}(B^{01}(e)+B^{(10)}(e)+B^{(11)}(e)-B^{(00)}(e)),
\]
 and summing this over $e_{0},e_{1},e_{2}$ and applying Lemma \ref{lem: B_ij-1}
yields 
\[
\E\en_{e_{0}}+\E\en_{e_{1}}+\E\en_{e_{2}}=\frac{1}{2Z^{I}}12(11-5\sqrt{3})c_{\tri}^{\frac{1-|\T^{\delta}|}{2}}\sqrt{\mathrm{det}^{\star}\Delta^{00}}\frac{1}{|\torus|}.
\]
Now, we can recall that 
\begin{multline*}
2Z^{I}=Z_{01}+Z_{10}+Z_{11}=\sqrt{\det\KW^{01}}+\sqrt{\det\KW^{10}}+\sqrt{\det\KW^{11}}\\
=c_{\tri}^{-\frac{|\torus|}{2}}\left(\sqrt{\det\Delta_{\delta}^{01}}+\sqrt{\det\Delta_{\delta}^{10}}+\sqrt{\det\Delta_{\delta}^{11}}\right),
\end{multline*}
and putting all together, after some tedious algebra to simplify the
constant in front, yields the result.
\end{proof}
\bibliographystyle{plain}
\bibliography{Biblio_of_Petri}

\end{document}